\documentclass[10pt,onecolumn]{IEEEtran}

\usepackage{setspace}
\usepackage{amsmath}
\usepackage{amsthm}
\usepackage{amssymb}
\usepackage{comment}
\usepackage{units}
\usepackage{graphicx}
\usepackage{float}
\usepackage{color}
\doublespace

\newtheorem{theorem}{Theorem}
\newtheorem{lemma}{Lemma}
\newtheorem{corollary}{Corollary}
\newtheorem{remark}{Remark}

\newtheorem{definition}{Definition}

 \def\bmu{{\pmb{\mu}}}
\def\b0{{\pmb{0}}} 

 \def\bb{{\mathbf{b}}} \def\bc{{\mathbf{c}}} 
 \def\bff{{\mathbf{f}}}

\def\bq{{\mathbf{q}}}   
 \def\bv{{\mathbf{v}}}

\def\bA{{\mathbf{A}}}  \def\bC{{\mathbf{C}}} 
 \def\bF{{\mathbf{F}}}  
\def\bI{{\mathbf{I}}}

\begin{document}
\title{Concatenated Coding Using Linear Schemes for Gaussian Broadcast Channels with Noisy Channel Output Feedback}
\author{Ziad Ahmad, ~\IEEEmembership{Student Member,~IEEE},
Zachary~Chance,~\IEEEmembership{Member,~IEEE}, and 
David~J.~Love,~\IEEEmembership{Senior Member,~IEEE}
\thanks{Z. Ahmad and D. J. Love are with the School of Electrical and Computer Engineering, Purdue University, West Lafayette, IN, USA. Z. Chance is with the MIT Lincoln Laboratory, Lexington, MA, USA.
}

}

\maketitle

\begin{abstract}
Linear coding schemes  have been the main choice of coding for the additive white Gaussian noise broadcast channel (AWGN-BC) with 
noiseless feedback in the literature. The achievable rate regions of these schemes
go well beyond the capacity region of the AWGN-BC without feedback. In this paper, a concatenating
coding design for the $K$-user AWGN-BC with noisy feedback is proposed that relies on linear feedback schemes to achieve rate tuples
outside the no-feedback capacity region. Specifically, a linear feedback code for the AWGN-BC with noisy feedback is used as an inner code that creates an effective  single-user channel from the transmitter to each of the receivers, and then open-loop coding is used for coding over these single-user channels. An achievable rate region of linear feedback schemes for noiseless feedback is shown to be achievable by the concatenated coding scheme for sufficiently small feedback noise level. Then, a linear feedback coding scheme for the $K$-user symmetric AWGN-BC
with noisy feedback is presented and optimized for use in the concatenated coding scheme. Lastly, we apply the concatenated coding design to the two-user AWGN-BC with a single noisy feedback link from one of the receivers.
\end{abstract}


\section{Introduction}
The demand for higher data rates in wireless communication systems continues to increase. However, there is concern that many of the popular approaches to physical layer design are only capable of minimal further enhancements \cite{dead}. In this paper, we look into one area that has not been fully explored which is the use of feedback in channel coding for increasing data rates.

The use of feedback in Gaussian channels dates back to the seminal paper by Schakwijk and Kailath (S-K)  \cite{Schal1}. Assuming a noiseless feedback link available from the receiver to the transmitter, the paper presented a simple linear scheme that achieves the capacity of the single-user additive white Gaussian noise (AWGN) channel. More importantly, the scheme has a probabilty of error that decays doubly exponentially with the blocklength as compared to at most linearly exponential decay for the same channel but without feedback \cite{shannon2}. The scheme was then extended by Ozarow \cite{Oz} to the AWGN broadcast channel (AWGN-BC), which is the focus of this paper, to show an improvement on the no-feedback capacity region using noiseless feedback. Also assuming noiseless feedback, the works in \cite{Kramer,El04,lqg} showed further improvements.

The only obstacle standing in the way of allowing these schemes to make it through to practical systems is the strong assumption of noiseless feedback. All of the beforementioned feedback coding schemes developed for the AWGN-BC with noiseless feedback are linear. For the point-to-point AWGN channel with feedback, it was shown in \cite{Kim2,ZaDa11}, that if the feedback noise level is larger than zero, no matter how low the level is, linear feedback schemes fail to achieve any positive rate. As we show in this paper, this negative result extends to the AWGN-BC.

Two recent works \cite{ramji,shayevitz2} presented achievable rate regions for the broadcast channel with general feedback. Both these regions where derived using schemes inspired by the example in \cite{Dueck}. In \cite{shayevitz2}, it is shown for two types of discrete memoryless channels that noisy feedback, specifically with sufficiently small feedback noise level, improves on the no-feedback capacity region. In \cite{ramji}, the achievable rate region is evaluated for the symmetric two-user AWGN-BC with a single feedback link from one of the receivers. In the high forward channel signal-to-noise ratio (SNR) regime, the scheme improves  on the no-feedback sum-capacity for a feedback noise level as high as the forward noise level. However, for low SNR (but still within practical values), the scheme's improvement over the no-feedback sum-capacity is negligible even for noiseless feedback.

In this paper, we consider the AWGN-BC with feedback. In particular, noiseless feedback will mean the transmitter has perfect access to the channel outputs in a causal fashion. On the other hand, noisy feedback will mean the transmitter has causal access to the channel outputs corrupted by AWGN in the feedback link from each receiver. We extend the concatenated coding scheme that was presented in \cite{ZaDa11} for the point-to-point AWGN with noisy feedback to the $K$-user AWGN-BC with noisy feedback.  Specifically, a linear feedback code for the AWGN-BC with noisy feedback is used as an inner code that creates an effective single-user channel from the transmitter to each of the receivers, and then open-loop (i.e., without feedback) coding is used for coding over these single-user channels. 

For the single-user case, the scheme in \cite{ZaDa11} showed improvements in error-exponents compared to the no-feedback case. For the AWGN-BC with noisy feedback, we use the extended concatenated coding scheme to show improvements on the no-feedback capacity region. The contributions and improvements on previous works will be stated towards the end of this section. Before that, we would like to comment on the practicality of the concatenated coding scheme presented in this paper. In fact, the concatenated coding scheme presented in this paper has the following attractive properties for practical systems:
\begin{itemize}
\item Feedback information is utilized using simple linear processing.
\item Open-loop coding is only used over single-user channels. Furthermore, when interference from the message points of other users is canceled out by the linear feedback code (as in the scheme of Section \ref{sect:linear-symmetric}), the effective single-user channels are pure AWGN channels for which open-loop codes are well developed in practice.
\item No broadcast channel coding techniques, like dirty paper coding or superposition coding, are required.
\end{itemize}

The results of Theorem \ref{thm:main}, Theorem \ref{thm:symm}, and Theorem \ref{thm:single-feedback} are for sufficiently small feedback noise levels (compared to forward noise levels). However, many broadcast communication systems can have small noise level over the feedback channels. This is especially true for systems where the receivers have a larger power available at their disposable than the transmitter. One example of such a system is found in satellite communications. In a satellite communcation system, the transmitter which is at the satellite would be broadcasting (possibly independent) data streams to different gateways present on earth. Satellites have much less power available than the gateways on earth. Another important application that possesses the same distribution of power is communication with implantable chips. In such an application, the chip implanted in the body of a human would like to broadcast different measurements to different devices that are located outside the body. Since the implantable chip powers itself from energy harvesting systems that convert ambient enegry to electrical energy, the transmitter would have a very small power available as compared to the receivers that are located outside the body. Therefore, assuming a low feedback noise level as compared to the foward noise level still captures many important applications that starve for improvement in rates or lower transmitter power consumption.  

The contributions of the paper can be summarized by the following:
\begin{itemize}
\item We show that if the feedback noise level for a receiver is strictly larger than zero, no matter how low the level is, linear feedback schemes can only achieve the zero rate to that receiver. This is an extension of the result derived in \cite{Kim2,ZaDa11} for the single-user case.
\item We extend the concatenated coding scheme presented in \cite{ZaDa11} to the $K$-user AWGN-BC with noisy feedback, and show an achievable rate region of linear feedback schemes to be achievable by the concatenated coding scheme for a sufficiently small feedback noise level. From this result, it is deduced that any achievable rate tuple by Ozarow's scheme \cite{Oz} for noiseless feedback can be achieved by the concatenated coding scheme for small enough feedback noise level.
\item We present a linear feedback scheme for the symmetric $K$-user AWGN-BC channel with noisy feedback that is optimized and used as an inner code in the concatenated coding scheme. For noiseless feedback, it is shown that the scheme achieves the same sum-rate as in \cite{lqg} but over the real channel, unlike the scheme presented in \cite{lqg} that requires a complex channel. We show that the latter sum-rate is also achievable for sufficienlty small feedback noise level. We also present achievable sum-rates versus feedback noise level otained using the same linear scheme in the design of the concatenated coding scheme.
\item We apply the concatenated coding idea to the two-user AWGN-BC with a single noisy feedback channel from one of the receivers. The scheme in \cite{bhaskaran} is used, with some modifications, as the inner code to show that any rate tuple that is achievable by the scheme in \cite{bhaskaran} for noiseless feedback can be achieved by concatenated coding for sufficiently small feedback noise level. This shows achievable rate tuples outside what is presented in \cite{ramji}, especially for low forward channel SNR.
\end{itemize}

The paper is organized as follows: In Section \ref{sect:linear-general}, we describe the channel setup and give a general framework for linear feedback coding. In Section \ref{sect:concatenated-general}, we present the concatenated coding scheme and its achievable rate region. In Section \ref{sect:linear-symmetric}, we present a linear feedback coding scheme for the symmetric AWGN-BC with noisy feedback that is utilized in the concatenated coding scheme in Section \ref{sect:concatenated-symmetric} for the same channel. In Section \ref{sect:single-feedback}, we present a concatenated coding design for the two-user AWGN-BC with one noisy feedback link from one of the receivers. The paper is concluded in Section \ref{sect:conclusion}.

\section{General Framework for Linear Feedback Coding} 
\label{sect:linear-general}
In this section, we formulate a general framework for linear feedback coding schemes for the $K$-user AWGN-BC with noisy feedback. 

\subsection{Channel Setup}
\label{sect:channel-setup}
We start by describing the channel setup that is depicted in Fig.~\ref{fig:channel}. The channel at hand has one transmitter and $K$ receivers. Before every block of transmission, the transmitter will have $K$ independent messages $W_1$, $W_2$, $\dots$, $W_K$, each to be conveyed reliably to the respective receiver.

After channel use $\ell$, the channel output at receiver $k$, for $k\in \mathbb K = \{1,2,\dots,K\}$, is given by
\begin{equation}
y_k[\ell] = x[\ell] + z_k[\ell],
\end{equation}
where $x[\ell] \in \mathbb R$ is the transmitted symbol at time $l$ and $\{z_k[\ell]\}$ are i.i.d. and such that $z_k[l] \sim \mathcal {N} (0,\sigma^2_{z_k})$. $z_k[\ell]$ is assumed independent of $x[\ell]$ for $k \in \mathbb K$. An average transmit power constraint, $P$, is imposed so that
\begin{equation}
\label{eq:pwrconstchannel}
E\left[\sum_{\ell=1}^{L}x^2[\ell] \right] \leq LP,
\end{equation}
where $L$ is the length of the transmission block.

 \begin{figure}[ht]
\centering
 \includegraphics{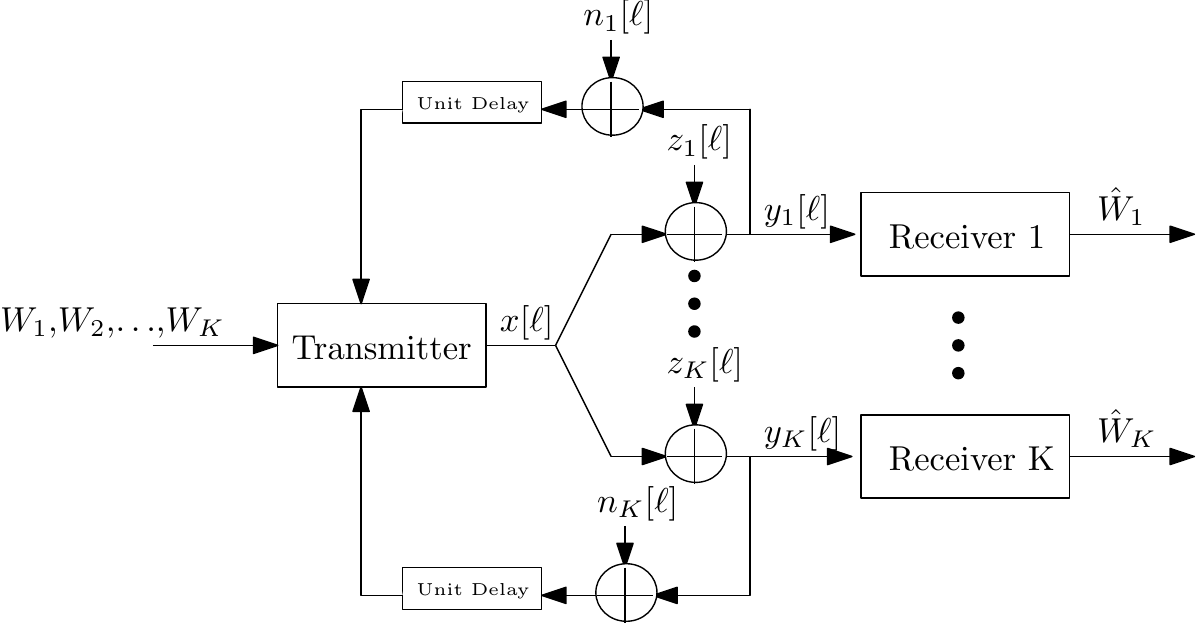}
\caption{AWGN-BC with feedback.}
\label{fig:channel}
\end{figure}

Through the presence of feedback links from each receiver to the transmitter, the transmitter will have access to noisy versions of the channel outputs of all receivers in a causal fashion. In particular, to form $x[\ell]$, the transmitter can use $\{y_1[1]+n_1[1],\dots,y_K[1]+n_K[1],\dots,y_1[\ell-1]+n_1[\ell-1],\dots,y_K[\ell-1]+n_K[\ell-1]\}$, where $\{n_k[\ell]\}$ are i.i.d. and such that $n_k[\ell] \sim \mathcal {N} (0,\sigma^2_{n_k})$. Since the transmitter knows what it had transmitted in the previous transmissions, it can subtract it and equivalently use $\{z_1[1]+n_1[1],\dots,z_K[1]+n_K[1],\dots,z_1[\ell-1]+n_1[\ell-1],\dots,z_K[\ell-1]+n_K[\ell-1]\}$. It is assumed that $n_k[\ell]$ is independent of a $x[\ell]$ for $k\in \mathbb K$, and $n_i[t]$ is independent of $z_j[s]$ for any $t,s \in \mathbb N$ and $i,j \in \mathbb K$. 

At the end of the transmission block, receiver $k$ will have an estimate of its message denoted by $\hat W_k$ for $k \in \mathbb K$.

\subsection{Linear Feedback Coding Framework}
\label{sect:linear-framework}
A general linear coding framework for the channel setup just described is presented next. Before each block of transmission, the transmitter maps each of the $K$ messages to a point in $\mathbb R$, which is termed a \emph{message point}. Specifically the message for the $k$-th receiver is mapped to $\theta_k \in \Theta_k \subseteq \mathbb R$ such that $|\Theta_k| = \lceil{2^{LR_k}}\rceil$, where $L$ is the length of the transmission block and $R_k$ is the rate of transmission for receiver $k$. 

Let $\mathbf x = [x[1], x[2],\dots,x[L]]^T$, $\mathbf z_k = [z_k[1], z_k[2],\dots,z_k[L]]^T$, $\mathbf n_k = [n_k[1], n_k[2],\dots,n_k[L]]^T$, and $\mathbf y_k = [y_k[1], y_k[2],\dots,y_k[L]]^T$, where the superscript $T$ denotes matrix transposition. Then we can write
\begin{equation*}
\mathbf x = \sum\limits_{k=1}^K\left[\mathbf g_k\theta_k+\mathbf F_k(\mathbf z_k + \mathbf n_k) \right],
\end{equation*}
where $\mathbf g_k \in \mathbb R^{L\times1}$ and  $\mathbf F_k \in \mathbb R^{L\times L}$ such that $\{\mathbf F_k\}$ are lower triangular matrices with zeros on the main diagonal so that casuality is ensured. 

The average transmit power constraint \eqref{eq:pwrconstchannel} can be written as
\begin{equation}
\label{eq:pwrconst}
E[\mathbf x^T\mathbf x] = \sum_{k=1}^K\mathbf g_k^T \mathbf g_k E[\theta_k^2]+ \sum\limits_{k=1}^K(\sigma_{z_k}^2+\sigma_{n_k}^2) \|\mathbf F_k\|_F^2 \leq LP.
\end{equation}
The received sequence at the $k$-th receiver can be written as
 \begin{equation*}
\mathbf  y_k = \mathbf x  + \mathbf z_k.
\end{equation*}

Each receiver will form an estimate of its message as a linear combination of its observed channel output sequence. Specifically, receiver $k$ will form an estimate $\hat\theta_k$ of $\theta_k$ as
\begin{equation*}
\hat\theta_k  = \mathbf q_k^T \mathbf y_k,
\end{equation*}
where $\mathbf q_k \in \mathbb R^{L\times1}$.

Breaking down $\hat\theta_k$ we have
\begin{equation}
\label{eq:thetahat}
\hat\theta_k = \mathbf q_k^T\mathbf g_k\theta_k + \sum\limits_{\substack{i=1\\i\neq k}}^K\mathbf q_k^T\mathbf g_i\theta_i + \sum\limits_{j=1}^K \mathbf q_k^T \mathbf F_j(\mathbf z_j + \mathbf n_j) +  \mathbf q_k^T\mathbf z_k.
\end{equation} 

\subsection{An Achievable Rate Region For Linear Feedback Coding} 
From \eqref{eq:thetahat}, any rate tuple $(R_1,\dots,R_K)$ that satisfies
\begin{equation}
\label{eq:userrate}
R_k<\lim_{L\rightarrow \infty}\frac{1}{2L}\log\left(1+SNR_k(L)\right),
\end{equation}
for all $k\in\mathbb K$ is achievable, where $SNR_k(L)$ is given in~\eqref{eq:usersnr}, and  in \eqref{eq:usersnr}, $\mathbf I$ is the identity matrix. 

\begin{equation}
\label{eq:usersnr}
SNR_k(L) = \frac{(\mathbf q_k^T\mathbf g_k)^2 E[\theta_k^2]}{\sum\limits_{\substack{i=1\\i\neq k}}^K\mathbf q_k^T \mathbf g_i E[\theta_i^2] + \sum\limits_{\substack{j=1\\j\neq k}}^K(\sigma_{z_j}^2+\sigma_{n_j}^2)\|\mathbf q_k^T\mathbf F_j\|^2 + \sigma_{z_k}^2\|\mathbf q_k^T(\mathbf I+ \mathbf F_k) \|^2 + \sigma_{n_k}^2\|\mathbf q_k^T\mathbf F_k \|^2}.
\end{equation}

Before closing this section, we show that for any linear feedback scheme, if the feedback noise variance of receiver $k$ is strictly greater than zero, i.e., if $\sigma_{n_k}^2>0$, then the only achievable rate for receiver $k$ is zero. This result is a direct extension of that of the single-user case shown in \cite{Kim2},\cite{ZaDa11}.
\begin{lemma}
\label{lemma:negative-result}
For any linear feedback scheme for the AWGN-BC with noisy feedback, if the feedback noise of receiver $k$ is strictly larger than zero, i.e., $\sigma_{n_k}^2>0$, then the only achievable rate $R_k$ for receiver $k$ is zero.
\end{lemma}
\begin{IEEEproof}
The result can be shown by direct extension of the single user result of \cite[Lemma 4]{ZaDa11}. 
We proceed by finding an upper bound on the achievable rates to receiver $k$ and show that it is equal to zero. First, removing the second term of \eqref{eq:thetahat}, we have 

\begin{equation}
\label{eq:thetahat2}
\hat\theta_k = \mathbf q_k^T\mathbf g_k\theta_k + \sum\limits_{j=1}^K \mathbf q_k^T \mathbf F_j(\mathbf z_j + \mathbf n_j) +  \mathbf q_k^T\mathbf z_k.
\end{equation}

Since the sum of the second and third terms in the right-hand side of \eqref{eq:thetahat2} is a Gaussian term, then any achievable rate $R_k$ to receiver $k$ must satisfy
\begin{equation*}
R_k \leq \lim_{L\rightarrow \infty}\frac{1}{2L}\log\left(1+\overline{SNR}_k(L)\right),
\end{equation*}
where $\overline{SNR}_k$ is the same as $SNR_k$ of \eqref{eq:usersnr} but with the term $\sum_{i\neq k}\mathbf q_k^T \mathbf g_i E[\theta_i^2]$
removed from the denominator. 

Now,
\begin{align*}
\overline{SNR}_k & \leq \frac{(\mathbf q_k^T\mathbf g_k)^2 E[\theta_k^2]}{\sigma_{z_k}^2\|\mathbf q_k^T(\mathbf I+ \mathbf F_k) \|^2 + \sigma_{n_k}^2\|\mathbf q_k^T\mathbf F_k \|^2}\\
& \leq \max \frac{(\mathbf q^T\mathbf g)^2 E[\theta_k^2]}{\sigma_{z_k}^2\|\mathbf q^T(\mathbf I+ \mathbf F) \|^2 + \sigma_{n_k}^2\|\mathbf q^T\mathbf F \|^2}\\
&\leq \frac{\sigma_{z_k}^2+\sigma_{n_k}^2}{\sigma_{n_k}^2}LP,
\end{align*}
where the maximization is over $\mathbf q,\mathbf F$, and $\mathbf g$ under the constraint  $\mathbf g^T \mathbf g E[\theta_k^2]+ (\sigma_{z_k}^2+\sigma_{n_k}^2) \|\mathbf F\|_F^2 \leq LP$, and the last inequality is by \cite[Lemma 3]{ZaDa11}. Then, if a rate $R_k$ is achievable to receiver $k$, it has to satisfy
\begin{align*}
R_k & \leq \lim_{L\rightarrow \infty}\frac{1}{2L}\log\left(1+\overline{SNR}_k(L)\right) \\
& \leq \lim_{L\rightarrow \infty}\frac{1}{2L}\log\left(1+\frac{\sigma_{z_k}^2+\sigma_{n_k}^2}{\sigma_{n_k}^2}LP\right) \\
& = 0.
\end{align*}
\end{IEEEproof}

\section{Concatenated Coding Scheme}
\label{sect:concatenated-general}
From Lemma \ref{lemma:negative-result}, we see that linear processing alone can only achieve the zero rate to the receiver with noisy feedback. Therefore, we need to do more than linear processing for noisy feedback in order to achieve positive rates, and possibly achieve rate tuples that are outside the no-feedback capacity region. We describe such a scheme in this section and that uses open-loop coding on top of linear processing to achieve rate tuples outside the no-feedback capacity region.
  
For any linear feedback code, we observe from \eqref{eq:thetahat} that for receiver $k$, the stochastic relation between $\theta_k$ and $\hat \theta_k$ can be modeled as a single-user channel without feedback, as in Fig.~\ref{fig:superchannel}. This channel will be termed $k$-th user \emph{superchannel}. Since we can perform open-loop coding over the superchannel for each user, we have converted the problem to single-user coding without feedback. This will be the main idea behind the concatenated coding scheme to be described in this section. We call the scheme a concatenated coding scheme because of the use of open-loop codes in concatenation with a linear feedback code that creates the superchannels, which shares many similarities to the definition in \cite{Forney} but here for a multi-user channel. Note that the time index $m$ in Fig.~\ref{fig:superchannel} is shown to indicate that the superchannel will be used more than once for open-loop coding. The time index $m$ will be defined later as we describe open-loop coding over the superchannels.

\begin{figure}[h]
\centering
 \includegraphics{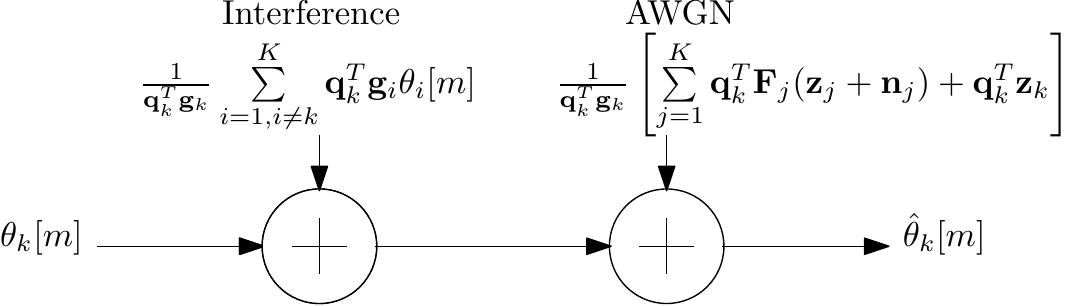}
\caption{Superchannel model.}
\label{fig:superchannel}
\end{figure}

Fig.~\ref{fig:concatenated} shows the overall concatenated coding scheme that will be described next. In each block of transmission, $K$ independent messages, $W_1$, $W_2$, $\dots$, $W_K$, will be available at the transmitter that are to be reliably coveyed, each to the respective receiver. The transmitter will use an open-loop code to encode each of the messages (i.e., will use $K$ open-loop encoders). All open-loop encoders use codebooks of equal blocklength $M$. Let the chosen codeword of the $k$-th open-loop encoder be $[\theta_k[1],\theta_k[2],\dots,\theta_k[M]]$. Similar to \cite{Forney} but for the AWGN-BC, we will term the block consisting of the $K$ open-loop encoders, which takes the $K$ messages as input and gives as an output $K$ coderwords each of length $M$, the \emph{outer code encoder}. At each time $m \in \{1,2,\dots,M\}$, the outer code encoder will have as an output, $\theta_1[m],\theta_2[m],\dots,\theta_K[m]$. 

 \begin{figure}[ht]
\centering
 \includegraphics{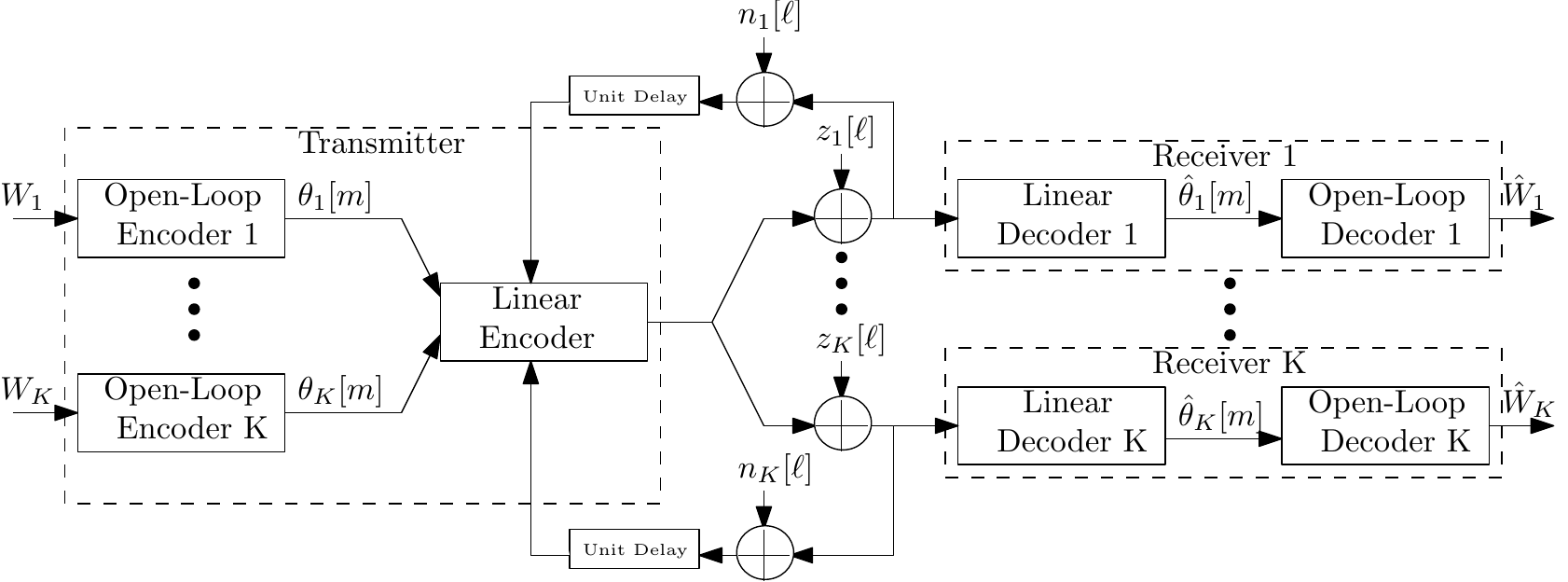}
\caption{Concatenated coding scheme.}
\label{fig:concatenated}
\end{figure}

For each set of $\theta_1[m],\theta_2[m],\dots,\theta_K[m]$, the transmitter will use a linear feedback code that will use the AWGN-BC with feedback $L$ times to have each receiver estimate its corresponding open-loop encoder output symbol, specifically, to have receiver $k$ estimate $\theta_k[m]$. The linear feedback code will be termed the \emph{inner code}. Its encoder will be termed the \emph{inner code encoder}, and its decoder at receiver $k$ will be termed the \emph{k-th inner code decoder}.  The $k$-th inner code decoder will output a linear estimate of $\theta_k[m]$. Let the estimate of $\theta_k[m]$, which is to be formed at receiver $k$, be $\hat\theta_k[m]$. 

Receiver $k$ will use an open-loop decoder, termed the \emph{k-th outer code decoder}, that corresponds to its open-loop encoder, to decode its message by observing the sequence $\hat\theta_k[1],\hat\theta_k[2],\dots,\hat\theta_k[M]$. 

The overall code for the AWGN-BC with feedback is of blocklength $ML$. Since for each $m$, the inner code encoder transmits with at most $LP$ of power, then the overall code uses a transmit power of at most $MLP$, and hence satisfies the codeword average power constraint. At receiver $k$, the SNRs is the same for all $\hat\theta_k[1],\hat\theta_k[2],\dots,\hat\theta_k[M]$, and is given by \eqref{eq:usersnr} if the time index is dropped (i.e., if $\theta_k[m]$ is simply written as $\theta_k$ for all $m$). Thus, if a linear code is fixed with blocklength $L$, the concatenated coding scheme described above can be designed to achieve any rate tuple $(R_1,R_2,\dots,R_k)$ that satisfies 
\begin{equation}
\label{eq:concatenatedrate}
R_k<\frac{1}{2L}\log\left(1+SNR_k(L)\right)
\end{equation}
for all $k\in \mathbb K$.

\begin{theorem}
\label{thm:main}
Given a linear feedback scheme over an AWGN-BC with noiseless feedback, for any rate tuple $(R_1,R_2,\dots,R_K)$ that satisfies \eqref{eq:userrate} for $k\in \mathbb K$, there exist $\epsilon_1>0,\dots,\epsilon_K>0$ such that the same rate tuple $(R_1,R_2,\dots,R_K)$ can be achieved by the concatenated coding scheme (scheme of Fig.~\ref{fig:concatenated}) over the same AWGN-BC but with $\sigma_{n_k}^2$ as large as $\epsilon_k$ for $k\in\mathbb K$.
\end{theorem}
\begin{IEEEproof}
For the given linear feedback coding scheme the SNR at receiver $k$ for blocklength $L$ is given by $SNR_k(L)$ of \eqref{eq:usersnr}. In this proof, we will make the dependence of the SNR on the blocklength and the feedback noise variances explicit, e.g., for a linear feedback code with blocklength $L$ that works according to the given linear feedback coding scheme over AWGN-BC with feedback noise variance for receiver $k$ of $\sigma^2_{n_k}$ will be written as $SNR_k(L,\sigma^2_{n_1},\dots,\sigma^2_{n_K})$. Note, here the dependence on $\sigma^2_{n_1},\dots,\sigma^2_{n_K}$ is just for the explicit values, i.e., if $\mathbf g_1$,$\dots$,$\mathbf g_K$, $\mathbf F_1$,$\dots$,$\mathbf F_K$, or $\mathbf q_1$,$\dots$,$\mathbf q_K$ depend on $\sigma^2_{n_1},\dots,\sigma^2_{n_K}$, it is not captured by the arguments of $SNR_k$.

For the given rate tuple $(R_1,R_2,\dots,R_K)$, assume $R_k>0$ for $k\in \mathbb K$; for the case of $R_k=0$ for some $k$, the proof below works the same but with trivially achieving the zero rates. Then,
\begin{equation*}
R_k<\lim_{L\rightarrow \infty}\frac{1}{2L}\log\left(1+SNR_k(L,0,\dots,0)\right),
\end{equation*}
for $k\in\mathbb K$.
Hence, there exists $L_0$ such that
\begin{equation*}
R_k<\frac{1}{2L_0}\log\left(1+SNR_k(L_0,0,\dots,0)\right)
\end{equation*}
for all $k\in \mathbb K$. Let the matrices of the given linear scheme for blocklength $L_0$ be $\mathbf g_1,\dots,\mathbf g_K$, $\mathbf F_1,\dots,\mathbf F_K$, and $\mathbf q_1,\dots,\mathbf q_K$ with the power constraint
\begin{equation*}
\sum_{k=1}^K\mathbf g_k^T \mathbf g_k E[\theta_k^2]+ \sum\limits_{k=1}^K\sigma_{z_k}^2 \|\mathbf F_k\|_F^2 \leq L_0P.
\end{equation*}

Let $g_{k1}$ be the first entry of $\mathbf g_k$ for $k\in \mathbb K$. Since $R_k>0$, then at least one entry of $\mathbf g_k$ is non-zero. Assume without loss of generality that $g_{k1}$ is non-zero. Also, assume that $g_{k1}>0$ (the proof still works in a similar way if $g_{k1}$ is assumed negative). For $k\in \mathbb K$, let $\mathbf g'_k$ be such that
\begin{equation}
\mathbf g'_{k}=\mathbf g_{k}-\epsilon'_k \left[\begin{array}{c} 1 \\ 0 \\ \vdots \\ 0 \end{array}\right],
\end{equation} 
where $g_{k1} - \epsilon'_k>0$ and $\epsilon'_1>0$, $\epsilon'_2>0$, \dots, $\epsilon'_K>0$ are to be chosen next. 

Choose $\epsilon'_1>0$,$\epsilon'_2>0$,\dots,$\epsilon'_K>0$ such that
\begin{equation*}
R_k<\frac{1}{2L_0}\log\left(1+SNR_k'(L_0,0,\dots,0)\right)
\end{equation*}
for all $k\in \mathbb K$, where $SNR_k'$ is the same function as $SNR_k$ but that uses $\mathbf g_1',\dots,\mathbf g_K'$ in place of $\mathbf g_1$,$\dots$,$\mathbf g_K$. This is possible by the continuity of $SNR_k$ at $\mathbf g_1$, $\mathbf g_2$, $\dots$, $\mathbf g_K$.

Now, choose  $\epsilon''_1>0$, $\epsilon''_2>0$,\dots, $\epsilon''_K>0$ such that
\begin{equation*}
\epsilon''_k \leq \frac{(\mathbf g_k^T \mathbf g_k -\mathbf {g'_k}^{T} \mathbf g'_k)E[\theta_k^2]}{ \|\mathbf F_k\|_F^2}
\end{equation*}
for $k \in \mathbb K$.
Also, choose $\epsilon'''_1>0$, $\epsilon'''_2>0$, $\dots$, $\epsilon'''_K>0$ such that
\begin{equation*}
R_k<\frac{1}{2L_0}\log\left(1+SNR_k'(L_0,\epsilon_1''',\dots,\epsilon_K''')\right)
\end{equation*}
for $k \in \mathbb K$.

Let $\epsilon_k = \min\{ \epsilon''_k,\epsilon'''_k \}$ for $k \in \mathbb K$.
Then, we have 
\begin{equation*}
\sum_{k=1}^K\mathbf {g'_k}^{T} \mathbf g_k' E[\theta_k^2]+ \sum\limits_{k=1}^K(\sigma_{z_k}^2+\epsilon_k) \|\mathbf F_k\|_F^2 \leq L_0P,
\end{equation*}
and
\begin{equation*}
R_k<\frac{1}{2L_0}\log\left(1+SNR_k'(L_0,\epsilon_1,\dots,\epsilon_K)\right)
\end{equation*}
for $k \in \mathbb K$.
Hence, for the same foward AWGN-BC but with feedback noise variances $\epsilon_1>0,\dots,\epsilon_K>0$, we have found a linear feedback code of blocklength $L_0$ defined by the matrices $\mathbf g_1',\dots,\mathbf g_K',\mathbf F_1,\dots,\mathbf F_K$, and $\mathbf q_1,\dots,\mathbf q_K$, that satisfies the power constraint, and that attains SNR at receiver $k$ of $SNR_k'(L_0,\epsilon_1,\dots,\epsilon_K)$ that is such that
\begin{equation*}
R_k<\frac{1}{2L_0}\log\left(1+SNR_k'(L_0,\epsilon_1,\dots,\epsilon_K)\right).
\end{equation*}
Using this linear code as an inner code, and by \eqref{eq:concatenatedrate}, the concatenated coding scheme achieves the rate tuple $(R_1,\dots,R_K)$.
\end{IEEEproof}

\begin{remark}
\label{remark:complex}
The result of Theorem~\ref{thm:main} can be directly extended to the complex AWGN-BC with complex AWGN feedback channels.
\end{remark}

\begin{remark}
\label{remark:ozarow}
In \cite{Oz}, the scheme is linear, and in addition to that, the achievable rate region presented in \cite{Oz} is the same as the set of rate tuples that satisfy \eqref{eq:userrate} for $k\in \mathbb K$. Hence, the achievable rate region for noiseless feedback in \cite{Oz} can be achieved by the concatenated coding scheme of Fig.~\ref{fig:concatenated} for sufficiently small feedback noise level. In  \cite{Oz}, an auxiliary Gaussian random variable $w$ is added to the first two transmissions, and only minor steps are needed to accomodate that in the proof of Theorem \ref{thm:main}.
\end{remark}

\section{A Linear Coding Scheme For The Symmetric AWGN-BC with Feedback}
\label{sect:linear-symmetric}
For designing the inner code of the concatenated coding scheme presented in Section~\ref{sect:concatenated-general}, we would ultimately like to find a linear coding scheme that maximizes the SNR at all receivers. However, to make the problem more tractable, we focus our attention on the symmetric case and impose some constraints on the scheme. 

With these constraints, and using the same channel setup of Section~\ref{sect:linear-general}, we present a linear coding scheme for the symmetric $K$-user AWGN-BC with feedback. Symmetric here means that all forward noises are of equal variances and all feedback noises are of equal variances too. Denote by $\sigma^2_z$ the forward noise variance and by $\sigma^2_n$ the feedback noise variance. We will set $\sigma^2_z \mathrel{\mathop:}=1$ so that $\sigma^2_n$ will represent the ratio $\nicefrac{\sigma^2_n}{\sigma^2_z}$ and $P$ will represent the channel SNR $\nicefrac{P}{\sigma^2_z}$ . The scheme we will develop will rely on techniques similar to code division multiple access (CDMA) techniques for nulling cross user interference. In this section, the total blocklength will be $L=\tilde{L}+K-1$, where $\tilde{L} \in \mathbb N$. The reason behind introducing a new parameter $\tilde{L}$ will be clearer as we describe the scheme. We assume that $K$ is an integer power of $2$, specifically $K \in \{2,4,8,16,\dots\}$.

Similar to the general formulation of Section~\ref{sect:linear-general}, the transmitter will map each of the independent $K$ messages to a message point in $\mathbb R$. Specifically, the transmitter maps the message intended to receiver $k$ to a point $\theta_k \in \Theta_k$ where $\Theta_k \subseteq \mathbb R$ and is such that $|\Theta_k|=\lceil{2^{LR_k}}\rceil$, where $R_k$ is the rate for receiver $k$. Similar to Ozarow's scheme~\cite{Oz}, the first $K$ transmissions are used to send the message points in an orthogonal fashion. We will assume that time division is used for achieving that and let $x[k]=\theta_k$ for $k\in \mathbb K$ (note that the traditional CDMA could be used too). The remaining $L-K$ transmissions will be used for sending feedback information in a CDMA-like manner that shares similarties to the techniques used in \cite{Kramer}. 

Let $\mathbf{\tilde z}_k = [z_k[k], z_k[K+1],\dots,z_k[L]]^T$, $\mathbf{\tilde n}_k = [n_k[k], n_k[K+1],\dots,n_k[L]]^T$, and $\mathbf{\tilde y}_k = [y_k[k], y_k[K+1],\dots,y_k[L]]^T$. Thus, we could write
\begin{equation}
\mathbf{\tilde y}_k = \mathbf{e_1}\theta_k+\sum_{k=1}^{K}\mathbf{\tilde F}_k(\mathbf{\tilde z}_k+\mathbf{\tilde n}_k) + \mathbf{\tilde z}_k,
\end{equation}
where $\mathbf{\tilde F}_k \in \mathbb R^{\tilde{L}\times\tilde{L}}$ and $\mathbf e_1$ is the first column of the $\tilde L \times \tilde L$ identity matrix.

For $k=1,2\dots,K$, let $\mathbf c_k \in \mathbb R^{1\times K}$ be of entries in $\{-1,1\}$ and such that 
\begin{equation*}
\bc_{i}^T\bc_{j} =
\begin{cases}
K, & i = j,\\
0, & i \neq j.
\end{cases}
\end{equation*}

\begin{remark}
\label{remark:K}
The vectors $\mathbf c_1,\dots,\mathbf c_K$ can be chosen as the columns of a $K\times K$ Hadamard matrix. For this reason, we have constrained $K$ to be an integer power of 2. Note, however, that if the channel at hand was complex, this constraint on $K$ can be alleviated by using complex Hadamard matrices, and all sum-rates derived for the real channel can be similarly achieved per real dimension over the complex channel for any $K \geq 2$.
\end{remark}

We will restrict $\mathbf{\tilde F}_k$ to be such that
\begin{equation*}
\mathbf{\tilde F}_k = \mathbf C_k \mathbf F,
\end{equation*}
where $\mathbf C_k \in \mathbb R^{\tilde L\times \tilde L}$ is such that
\begin{equation}
\label{eq:Cmatrix}
\mathbf [\mathbf C_k]_{ij}=
\begin{cases}
c_k[i\bmod K], & i = j,\\
0, & i \neq j,
\end{cases}
\end{equation}
and $\mathbf F \in \mathbb R^{\tilde L \times \tilde L}$ is a lower triangular matrix with zeros on the main diagonal to ensure causality and whose consrtuction will be described later.

With $\mathbf x$ defined as in Section \ref{sect:linear-general}, the average transmit power is bounded by
\begin{equation}
 E[\mathbf{x}^T\mathbf{x}] \leq LP. \label{power_multi}
\end{equation}
\noindent The power budget (\ref{power_multi}) can be divided between two different quantities: the power dedicated to the messages and the power used for feedback encoding.  This can be seen by expanding out (\ref{power_multi}) as
\begin{equation}
E[\mathbf x^T \mathbf x] = \displaystyle\sum_{k = 1}^{K}E[\theta_{k}^2] + K(1 + \sigma_{n}^2)\|\bF\|_{F}^{2}.
\end{equation}
The first quantity on the right hand side, $\sum_{k = 1}^{K}E[\theta_{k}^2]$, can be seen as the power used for transmitting the messages while the second term, $K(1 + \sigma_{n}^2)\|\bF\|_{F}^{2}$, is interpreted as the power utilized for transmitting feedback information.  Due to this trade-off, a new parameter $\gamma \in [0,1]$ is introduced such that
\begin{equation}
\displaystyle\sum_{k = 1}^{K}E[\theta_{k}^2] = (1-\gamma)LP,
\end{equation}
\noindent and
\begin{equation}
\label{eq:feedback-power}
K(1 + \sigma_{n}^2)\|\bF\|_{F}^{2} \leq \gamma LP.
\end{equation}
\noindent Thus, $\gamma$ can be thought of as the normalized ratio of power spent on encoding feedback information. Since the channel is symmetric, we will assume that 
\begin{equation*}
E[\theta_k^2] = \frac{1}{K}(1-\gamma)(\tilde L+K-1)P
\end{equation*}
for all users. 


The receiver creates its estimate, $\hat{\theta}_{k}$ as
\begin{equation}
\hat{\theta}_{k} = \bq^T\bC_{k}\mathbf{\tilde y_{k}},
\end{equation}
\noindent where $\bq \in \mathbb{R}^{\tilde L}$. 

\begin{equation}
SNR_{k}(\tilde L) = \frac{(q[1])^2\frac{1}{K}(1-\gamma)(\tilde L+K-1)\rho}{\|\bq^T(\bI + \bF)\|^2 + \sigma_{n}^2\|\bq^T\bF\|^2 + (1 + \sigma_{n}^2)\displaystyle\sum_{\substack{i=1 \\ i \neq k}}^{K}\|\bq^T\bC_{k}\bC_{i}\bF\|^2}.\label{SNR_multi}
\end{equation}

Then the received SNR for the $k$-th receiver is given by \eqref{SNR_multi}.

\begin{definition} A sum-rate $R$ is said to be achievable if there exists a rate tuple $(R_1,R_2,\dots,R_k)$ that is achievable and satisfying
\begin{equation}
R = \displaystyle\sum_{i = 1}^{K}R_{k}.
\end{equation}
\end{definition}

\noindent Hence, any sum-rate $R$ that satisfies
\begin{equation}
R < \lim_{L \rightarrow \infty}\displaystyle\sum_{i = 1}^{K}\frac{1}{2L}\log\left(1 + SNR_{k}(L)\right),\label{sumrate}
\end{equation}
\noindent is achievable where $SNR_{k}(N)$ is written to show the dependence of the received SNR on the blocklength.

\subsection{Interference Nulling}
We will constraint our scheme to satisfy
\begin{equation}
\label{eq:ortho-requirement}
\sum_{\substack{i=1\\i \neq k}}^K\|\bq^T\bC_{k}\bC_{i}\bF\|^2 = 0,
\end{equation}
so that cross user interference is nulled to zero.

In the following lemma, constraints on the transmission scheme are given to satisfy requirement \eqref{eq:ortho-requirement}.
\begin{lemma}\label{lemma_isit_1}
Let $\mathbf C_k$ be defined as in \eqref{eq:Cmatrix} for $k=1,2,\dots,K$. Then, the following forms of $\mathbf q$ and $\mathbf F$ satisfy \eqref{eq:ortho-requirement}\emph{:}
\begin{itemize}
    \item For a real number $\beta \in (0,1)$
    \[
    \bq = \left[1,\beta^2,\beta^4,\dots,\beta^{2(\tilde L-1)}\right]^T.
    \]
    \item Let 
\begin{equation}
\mathbf f = \left[1,\beta^{-2},\beta^{-4},\dots,\beta^{-2(\tilde K-1)}\right]^T.
\end{equation}
The $i^{th}$ column of the $\bF$ matrix is built by $\lfloor\frac{\tilde L-i}{K}\rfloor$ scaled copies of $\mathbf f$ below the main diagonal and the remaining entries are set to zero. The scaling coefficient for the $i^{th}$ column and the $j^{th}$ copy of $\mathbf f$ will be called $\mu_{i,j} \in \mathbb{R}$.
Specifically, the  $i^{th}$ column of the $\bF$ matrix is given by
\begin{equation}
[
\underbrace{0\text{\space}\dots \text{\space}0}_\textrm{$\scriptstyle{i}$} \text{\space} \mu_{i,1} \mathbf f^T \text{\space} \text{\space} \mu_{i,2} \mathbf f^T \text{\space} \dots \text{\space}  \mu_{i,\lfloor\frac{\tilde L-i}{K}\rfloor} \mathbf f^T \text{\space} \underbrace{0\text{\space} \dots \text{\space} 0}_\textrm{$\scriptstyle{\tilde L-i-K\lfloor\frac{\tilde L-i}{K}\rfloor}$} 
 ]^T.
\end{equation}

\end{itemize}
\end{lemma}
\begin{proof}
The form of $\bF$ stems from
the following observation:  For any $\bv_{1} \in \mathbb{R}^{\tilde L}$ and
$\bv_{2} \in \mathbb{R}^{\tilde L}$, to satisfy
\[
\bv_{1}^T\bC_{i}\bC_{j}\bv_{2} =
\begin{cases}
\bv_{1}^T\bv_{2},& i = j\\
0,& i \neq j\\
\end{cases}
\]
\noindent the vectors $\bv_{1}$ and $\bv_{2}$ can be constructed as $v_{2}[i] = \frac{1}{v_{1}[i]}$ for all $i = 1,2,\ldots,\tilde L$.
Using this fact and the condition that it must hold between $\bq$ and $K$ shifts of $\bff$, the lemma is constructed.  The further choice that $\beta \in (0,1)$ is
to keep the norm of $\bq$ bounded as $\tilde L \rightarrow \infty$. Note that $\mathbf F$ is all zeros for $\tilde L\leq K$.
\end{proof}

\subsection{SNR Optimization}

 With $\mathbf q$ and $\mathbf F$ having forms as in Lemma \ref{lemma_isit_1}, the SNR at any of the receivers can be written as
\begin{equation}
\label{eq:ortho-SNR}
SNR(\tilde L) = \frac{\frac{1}{K}(1-\gamma)(\tilde L+K-1)P}{\|\bq^T(\bI + \bF)\|^2 + \sigma_{n}^2\|\bq^T\bF\|^2}.
\end{equation} 
In the following lemma, given $\gamma$ and $\beta$, we optimize $SNR$ \eqref{eq:ortho-SNR} over the values of $\mu_{i,j}$.

\begin{lemma}\label{lemma_isit_2}
Assume $\tilde L>K$. Given $\gamma,\beta \in (0,1)$ and following the forms of $\mathbf q$ and $\mathbf F$ as in Lemma \ref{lemma_isit_1}, the $\mu_{i,j}$
values of  $\mathbf F$ that maximize the received SNR \eqref{eq:ortho-SNR} given the power constraint \eqref{power_multi} can be obtained as follows: \\
\begin{enumerate}
\item Define
\[
\bmu_{i} = \left[\mu_{i,1},\mu_{i,2},\ldots,\mu_{i,\left\lfloor {\frac{\tilde L-i}{K} }\right\rfloor}\right]^T,
\]
\vspace{-4mm}
\[
\bv_{i} = K\beta^{i-1}\left[1, \beta^K, \ldots, \beta^{K(\left\lfloor {\frac{\tilde L-i}{K} }\right\rfloor-1)}\right]^T,
\]
\noindent  for $i = 1,2,\ldots,\tilde L-K$.
\item Then, the $\bmu_{i}$ that maximize the received SNR are constructed as
\[
\bmu_{i} = -\frac{q_{i}}{(1+\sigma_{n}^2)\|\bv_{i}\|^2 + \lambda}\bv_{i},
\]
\noindent where $\lambda \geq 0$ is chosen to satisfy
\[
\displaystyle\sum_{i = 1}^{\tilde L-K}\|\bmu_{i}\|^2 \leq \frac{\gamma L P}{K(1 + \sigma_n^2)\|\bff\|^2}.
\]
\end{enumerate}
\end{lemma}
\begin{proof}
With the definitions in Lemma \ref{lemma_isit_2}, the denominator of the received SNR in \eqref{eq:ortho-SNR} can be rewritten as
\begin{equation}
\displaystyle\sum_{i = \tilde L-K+1}^{\tilde L}q_{i}^2+\displaystyle\sum_{i = 1}^{\tilde L-K}\left(q_{i} + \bv_{i}^T\bmu_{i}\right)^2 + \sigma_n^2\displaystyle\sum_{i = 1}^{\tilde L-K}\left(\bv_{i}^T\bmu_{i}\right)^2\label{minobj}.
\end{equation}
\noindent Then, it can be shown that to minimize (\ref{minobj}), one should let $\bmu_i = -b_{i}\frac{\bv_{i}}{\|\bv_{i}\|}$ for some scalars $b_{i}$ for $i = 1,2,\ldots,\tilde L-K$.
The sum of the second and third terms of \eqref{minobj} can now be rewritten as
\begin{equation}
\|\bA\bb - \bq\|^2 + \sigma_{n}^2\|\bA\bb\|^2,\label{obj_multi}
\end{equation}
\noindent where $\bA \in \mathbb{R}^{\tilde L \times \tilde L-K}$ is
\[
\bA = \left[\begin{array}{ccccc}
\|\bv_{1}\| & 0 & 0 & \cdots & 0\\
0 & \|\bv_{2}\| & 0 & \cdots & 0\\
\vdots & & \ddots & & \vdots\\
0 & 0 & \cdots & 0 & \|\bv_{\tilde L-K}\|\\
0 & 0 & \cdots & & 0\\
\vdots & &\vdots & &\vdots\\
0 & 0 & \cdots & & 0\\
\end{array}\right]
\]
\noindent and $\bb = [b_{1},b_{2},\ldots,b_{\tilde L-K}]^T$.  To minimize
(\ref{obj_multi}) and abide by the average power constraint, we use
Lagrange multipliers to obtain the $\bb$ that minimizes (\ref{obj_multi}) is
\begin{equation}
\bb_{min} = \left[(1 + \sigma_n^2)\bA^T\bA + \lambda\bI\right]^{-1}\bA^T\bq,
\end{equation}
\noindent where $\lambda$ is chosen to satisfy the power constraint.  Thus, using $\bb_{min}$ to build $\bmu_{i}$, we produce the lemma.
\end{proof}

The optimal form of $\bmu_i$ in Lemma \ref{lemma_isit_2} depends on $\lambda$ for which a closed form is generally hard to obtain. We will leave the optimal form for numerical optimization. However, notice that $\lambda \rightarrow 0$ as $L \rightarrow \infty$ in which case it can be shown that 
\begin{equation}
\mu_{i,j} = -\frac{1-\beta^{2K}}{(1+\sigma_n^2)K}\beta^{K(j-1)}.
\end{equation}

Furthermore, as $\sigma_n^2 \rightarrow 0$, we have 
\begin{equation}
\label{eq:mu}
\mu_{i,j} = -\frac{1-\beta^{2K}}{K}\beta^{K(j-1)}.
\end{equation} 
Using \eqref{eq:mu}, for $\tilde L >K$ the SNR at any of the receivers can be written as
\begin{equation}
\label{eq:snr-symm-scheme}
SNR(\tilde L) = \frac{\frac{1}{K}(1-\gamma)(\tilde L+K-1)P}{g(\tilde L,\beta)+\sigma_n^2h(\tilde L,\beta)},
\end{equation}
where
\begin{equation*}
g(\tilde L,\beta) = \sum^{\tilde L}_{i=\tilde N-K+1}\beta^{2(i-1)} + \sum^{\tilde L-K}_{i=1}\beta^{\left[2(i-1)+4K\left\lfloor{\frac{\tilde L-i}{K}}\right\rfloor\right]},
\end{equation*}
and
\begin{equation*}
h(\tilde L,\beta) = \sum_{i=1}^{\tilde L-K}\beta^{2(i-1)}\left(1-\beta^{2K\left\lfloor{\frac{\tilde L-i}{K}}\right\rfloor}\right)^2,
\end{equation*}
and the power constraint \eqref{eq:feedback-power} can be written as
\begin{equation}
\label{eq:powerconstraint}
e(\tilde L,\beta) \leq \frac{\gamma(\tilde L+K-1)P}{K(1+\sigma_n^2)},
\end{equation}
where 
\begin{equation*}
e(\tilde L,\beta) =\frac{(1-\beta^{2K})^2}{K^2(1-\beta^2)\beta^{2K}}\left(\tilde L-K-\sum^{\tilde L-K}_{i=1}\beta^{2K\left\lfloor{\frac{\tilde L-i}{K}}\right\rfloor}\right).
\end{equation*}
Then \eqref{sumrate} can be written as
\begin{equation}
\label{eq:sum-rate-scheme}
R<\lim_{\tilde L\rightarrow \infty} \frac{K}{2(\tilde L+K-1)}\log\left(1+SNR(\tilde L)\right).
\end{equation}
  
In the next lemma, we find upper and lower bounds on $SNR(\tilde L)$.
\begin{lemma}
\label{lemma:snrbound}
Assume $\sigma_n^2 = 0$. Then, $SNR(\tilde L)$ can be bounded as
\begin{equation*}
SNR_{lb}(\tilde L) \leq SNR(\tilde L) \leq SNR_{ub}(\tilde L),
\end{equation*}
where
\begin{equation*}
SNR_{lb}(\tilde L) = \frac{a_{lb}(1-\gamma)(\tilde L+K-1)\frac{P}{K}}{\beta^{2\tilde L}},
\end{equation*}
\begin{equation*}
SNR_{ub}(\tilde L) = \frac{(1-\beta^2)(1-\gamma)(\tilde L+K-1)\frac{P}{K}}{\beta^{2(\tilde L-K)}-\beta^{2\tilde L} + \beta^{2(\tilde L-K-1)}(1-\beta^{2(\tilde L-K)})},
\end{equation*}
and
\begin{equation*}
a_{lb} = \frac{(1-\beta^2)}{\beta^{-2K}(1+\beta^2)-1}.
\end{equation*}
Also, for large $\tilde L$
\begin{equation}
SNR(\tilde L) \approx SNR_{lb}(\tilde L) \approx SNR_{ub}(\tilde L).
\end{equation}
\end{lemma}

\begin{IEEEproof}
The second term of $g(\tilde L,\beta)$ can be upper bounded as
\begin{align*}
\sum^{\tilde L-K}_{i=1}\beta^{\left[2(i-1)+4K\left\lfloor{\frac{\tilde L-i}{K}}\right\rfloor\right]} & \leq \sum^{\tilde L-K}_{i=1}\beta^{\left[2(i-1) + 4(\tilde L-i-K+1)\right]} \\
& = \beta^{2(\tilde L-K+1)}\frac{1-\beta^{2(\tilde L-K)}}{1-\beta^2} \\
& \leq \beta^{2\tilde L}\frac{\beta^{-2(K-1)}}{1-\beta^2},
\end{align*}
where the first inequality is due to the fact that
\begin{equation}
 \left\lfloor{\frac{\tilde L-i}{K}}\right\rfloor \geq \frac{\tilde L-i-K+1}{K}.
\end{equation}
Using this bound, $SNR_{lb}$ can be reached.

On the other hand, the second term of $g(\tilde L,\beta)$ can be lower bounded as
\begin{align*}
\sum^{\tilde L-K}_{i=1}\beta^{\left[2(i-1)+4K\left\lfloor{\frac{\tilde L-i}{K}}\right\rfloor\right]} & \geq \sum^{\tilde L-K}_{i=1}\beta^{\left[2(i-1) + 4(\tilde L-i+K)\right]} \\
& = \beta^{2(\tilde L-K-1)}\sum^{\tilde L-K}_{i=1}\beta^{2(\tilde L -K -i)} \\
& = \beta^{2(\tilde L-K-1)}\frac{1-\beta^{2(\tilde L-K)}}{1-\beta^2}
\end{align*}
where the first inequality is due to the fact that
\begin{equation}
 \left\lfloor{\frac{\tilde L-i}{K}}\right\rfloor \leq \frac{\tilde L-i}{K}+1.
\end{equation}
Using this bound, $SNR_{ub}$ can be reached.

For large $\tilde L$, we can see that  $SNR_{lb}(\tilde L) \approx SNR_{ub}(\tilde L)$ and thus $SNR(\tilde L) \approx SNR_{lb}(\tilde L) \approx SNR_{ub}(\tilde L)$.
\end{IEEEproof}

\subsection{Achievable Sum-Rate For Noiseless Feedback}

For the noiseless feedback case (i.e., for $\sigma_n^2=0$), from Lemma \ref{lemma:snrbound}, we see that
\begin{equation*}
\lim_{\tilde L \rightarrow \infty} \frac{K}{2(\tilde L+K-1)} \log\left(1+SNR_{lb}(\tilde L)\right) =
	\lim_{\tilde L \rightarrow \infty}  \frac{K}{2(\tilde L+K-1)}  \log\left(1+SNR_{ub}(\tilde L)\right) = -K\log(\beta),
\end{equation*}

and hence 
\begin{equation*}
\lim_{\tilde L \rightarrow \infty} \frac{K}{2(\tilde L+K-1)} \log\left(1+SNR(\tilde L)\right) = -K\log(\beta).
\end{equation*}
Thus, any sum-rate $R$ is achievable if
\begin{equation}
\label{eq:sum-rate-beta}
R<  -K\log(\beta).
\end{equation}
In the following lemma, we show that $\beta$ and $\gamma$ can in fact be chosen so that the right-hand side of \eqref{eq:sum-rate-beta} is equal to the linear-feedback sum-rate bound derived in \cite{lqg}.

\begin{lemma}
\label{lemma:noiseless-sum-rate}
Let $\phi \in [1,K]$ be the solution of
\begin{equation}
\label{eq:phi-equation}
\left(1+P\phi\right)^{K-1} - \left[1+\frac{P}{K}\phi(K-\phi)\right]^K = 0. 
\end{equation}
The power constraint allows $\beta$ to be chosen as
\begin{equation*}
\beta^{-2K} = 1+P\phi
\end{equation*}
so that the scheme achieves any sum-rate $R$ satisfying
\begin{equation}
R<\frac{1}{2}\log\left(1+P\phi\right).
\end{equation}
\end{lemma}
\begin{IEEEproof}
Choose $\gamma=\frac{L-1}{L}$.
We choose $\beta$ such that all available power is consumed. Specifically, we choose $\beta$ such that
\begin{equation*}
\lim_{\tilde L \rightarrow \infty} \frac{e(\tilde L,\beta)}{\gamma(\tilde L+K-1)} = \frac{P}{K}.
\end{equation*}
The left-hand side of the above equation is equal to $\frac{(1-\beta^{2K})^2}{K^2(1-\beta^2)\beta^{2K}}$. Let $\beta^{-2K} = 1+P\phi$ and solve for $\phi$ instead of $\beta$. The resulting equation in $\phi$ can be reduced to \eqref{eq:phi-equation}. By \eqref{eq:sum-rate-beta}, the proof is complete.
\end{IEEEproof}

The sum-rate achieved here is the same as in \cite{lqg}. However, in \cite{lqg} the scheme requires a complex channel in order to achieve, per real dimension, the same sum-rate of Lemma \ref{lemma:noiseless-sum-rate}. This is especially true for $K>2$. Note, however, that the number of users $K$ is constrained to be an integer power of 2 for the real channel case.

\section{Concatenated Coding for the Symmetric AWGN-BC with Noisy Feedback}

\label{sect:concatenated-symmetric}
In this section, we consider the same concatenated scheme that was described in Section \ref{sect:concatenated-general}, but that relies on the linear scheme of Section \ref{sect:linear-symmetric} for coding over the symmetric AWGN-BC with noisy feedback. From Section \ref{sect:concatenated-general} and by the symmetry of the channel and scheme, if we fix a linear code of blocklength $L$ that works according to the scheme described in Section \ref{sect:linear-symmetric}, then any sum rate, $R$, can be achieved by the concatenated scheme just described if
\begin{equation}
\label{eq:achievable-sum-rate-concatenated}
R<\frac{K}{2(\tilde L +K -1)}\log\left(1+SNR(\tilde L)\right),
\end{equation}
where $SNR(\tilde L)$ is defined by \eqref{eq:ortho-SNR}.

\subsection{Achievable Sum-Rates For Small Enough Feedback Noise Level}
In this section, we discuss the achievable sum-rates for small enough feedback noise variance. From Theorem \ref{thm:main}, we know that what is achieved for the noiseless feedback case in Lemma \ref{lemma:noiseless-sum-rate} can be achieved for small enough feedback noise level by the concatenated coding scheme. However, for sum-rates close to the bound in Lemma \ref{lemma:noiseless-sum-rate}, the required inner code blocklength will be larger, and together with small $\sigma_n^2$, makes the choice of $\mu_{i,j}$ in \eqref{eq:mu} approximately optimal. For such case, and given a value for $\gamma$, Lemma \ref{lemma:powerconstraint} and Lemma \ref{lemma:f} will be useful for choosing the value of $\beta$. We will also use those lemmas to rederive the result of Theorem \ref{thm:main} but using the specifics of the scheme of this section.

\begin{lemma}
\label{lemma:powerconstraint}
$\beta$ that satisfies
\begin{equation}
\label{eq:strick-power-constraint}
\frac{(1-\beta^{2K})^2}{K(1-\beta^2)\beta^{2K}} \leq \frac{\gamma P}{1+\sigma_n^2},
\end{equation}
satisfies the power constraint \eqref{eq:powerconstraint} for any $\tilde L$.
\end{lemma}
\begin{proof}
$e(\tilde L,\beta)$ of \eqref{eq:powerconstraint} can be upper bounded as follows
\begin{equation*}
e(\tilde L,\beta) \leq \frac{(1-\beta^{2K})^2}{K^2(1-\beta^2)\beta^{2K}}(\tilde L+K-1).
\end{equation*}
Hence, $\beta$ that satisfies 
\begin{equation*}
\frac{(1-\beta^{2K})^2}{K^2(1-\beta^2)\beta^{2K}}(\tilde L+K-1) \leq \frac{\gamma(\tilde L+K-1)P}{K(1+\sigma_n^2)}
\end{equation*}
satisfies \eqref{eq:powerconstraint}.
\end{proof}

Note that for large $\tilde L$, the power lost by assuming the power constraint \eqref{eq:strick-power-constraint} instead of \eqref{eq:powerconstraint} becomes negligible.

\begin{lemma}
\label{lemma:f}
Let $f(\beta) = \frac{(1-\beta^{2K})^2}{K(1-\beta^2)\beta^{2K}}$. Then
\begin{itemize}
\item $f$ is a decreasing positive function on $(0,1)$. Specifically, if $\beta_1,\beta_2\in(0,1)$ are such that $\beta_1<\beta_2$, then $0<f(\beta_2)<f(\beta_1)$.
\item $f$ is a bijective function from $(0,1)$ to $(0,\infty)$.
\end{itemize}
\end{lemma}
\begin{proof}
Let $f'$ denote the first derivative of $f$ with respect to $\beta$. It can be shown that $f'(\beta)<0$ for $\beta \in (0,1)$ if and only if $p(x)>0$ for $x \in (0,1)$, where $p(x) = (1-K)x^{K+1}+Kx^K-(K+1)x+K$.
Now, let $p'$ and $p''$ denote the first and the second derivatives of $p$ with respect to $x$, respectively. To show that $p(x)>0$ for $x\in(0,1)$, we will use the fact that $p(1)=0$ and show that $p(x)$ is strictly decreasing on $(0,1]$.
We have,
\begin{equation*}
p'(x) = (1-K)(K+1)x^K+K^2x^{K-1}-(K+1)
\end{equation*}
and
\begin{equation*}
p''(x) = x^{K-2}K(K-1)\left[K-(K+1)x\right].
\end{equation*}
From $p''(x)$, we notice that $p'(x)$ is strictly increasing for $x \in (0,\frac{K}{K+1})$ and is strictly decreasing for $x\in (\frac{K}{K+1},1]$, and hence its maximum value on $(0,1]$ is at $x=\frac{K}{K+1}$. Hence for $x \in (0,1]$,
\begin{align*}
p'(x) &\leq p'\left(\frac{K}{K+1}\right) \\
	&=K\left(\frac{K}{K+1}\right)^{K-1}-(K+1) < 0.
\end{align*}
Therefore, $p(x)$ is a strictly decreasing function on $(0,1]$. But since $p(1)=0$, then $p(x)>0$ for $x\in(0,1)$.
So far, we have shown that $f$ is a strictly deceasing function on $(0,1)$. Now, since $f$ is a continous function on $(0,1)$ and since $\lim_{\beta \rightarrow 0}f(\beta) = \infty$ and $\lim_{\beta \rightarrow 1}f(\beta) = 0$, then $f((0,1)) = (0,\infty)$, and hence the proof is complete. 
\end{proof}

\begin{theorem}
\label{thm:symm}
For any sum-rate $R<\frac{1}{2}\log\left(1+P\phi\right)$, where $\phi$ is as defined in Lemma \ref{lemma:noiseless-sum-rate}, there exists $\epsilon>0$ such that the same sum-rate $R$ can be achieved by the concatenated coding scheme but with $\sigma_n^2$ as large as $\epsilon$.
\end{theorem}
\begin{IEEEproof}
For $R=0$, the proof is trivial. For $R>0$, choose $\gamma$ large enough such that $\frac{1}{2}\log\left(1+P\gamma\phi\right)>R$, where $\phi \in [1,K]$ is the solution of
\begin{equation*}
 \left(1+P\gamma\phi\right)^{K-1} - \left[1+\frac{P\gamma}{K}\phi(K-\phi)\right]^K = 0.
\end{equation*}
This allows us to choose $\beta \in [0,1]$ such that $-K\log(\beta)>R$ and $f(\beta)\leq P\gamma$. Choose, $\beta_0 \in [0,1]>\beta$ such that $-K\log(\beta)>-K\log(\beta_0)>R$. By Lemma~\ref{lemma:f}, there exists $\epsilon_1>0$ such that
\begin{equation*}
f(\beta_0)\leq\frac{P\gamma}{1+\epsilon_1}.
\end{equation*}
Define
\begin{equation*}
\tilde R(\tilde L,\sigma_n^2) = \frac{K}{2(\tilde L +K-1)}\log\left(1+SNR(\tilde L,\sigma_n^2)\right),
\end{equation*}
where $SNR(\tilde L,\sigma_n^2)$ here is given by
\begin{equation*}
SNR(\tilde L,\sigma_n^2) = \frac{\frac{1}{K}(1-\gamma)(\tilde L+K-1)P}{g(\tilde L,\beta_0)+\sigma_n^2h(\tilde L,\beta_0)}.
\end{equation*}
Since $\lim_{\tilde L \rightarrow \infty} \tilde R(\tilde L,0) = -k\log(\beta_0)>R$, there exists $\tilde L_0$ such that
\begin{equation*}
\tilde R(\tilde L_0,0) > R.
\end{equation*}
There also exists $\epsilon_2>0$ such that
\begin{equation*}
\tilde R(\tilde L_0,\epsilon_2) > R.
\end{equation*}
Let $\epsilon = \min\{\epsilon_1,\epsilon_2\}$. Since $\tilde R(\tilde L_0,\epsilon)\geq \max\{\tilde R(\tilde L_0,\epsilon_1),\tilde R(\tilde L_0,\epsilon_2)\}$ and since $f(\beta_0)\leq\frac{P\gamma}{1+\epsilon}$, by \eqref{eq:achievable-sum-rate-concatenated} and by Lemma~\ref{lemma:powerconstraint}, we have found $\gamma$, $\beta_0$, and $\tilde L_0$ such that the concatenated coding scheme achieves any sum-rate below $\tilde R(\tilde L_0,\epsilon)>R$ for feedback noise variance as large as $\epsilon$. Hence, $R$ is achieved.
\end{IEEEproof}

\subsection{Inner Code Blocklength}
In this section, we find an upper bound on the inner code blocklength required for the concatenated coding scheme to start achieving a certain sum-rate above the no-feedback sum-capacity. To do that, we assume noiseless feedback and make use of the $SNR$ lower bound in Lemma \ref{lemma:snrbound} and of Lemma \ref{lemma:powerconstraint}. For sum-rates close to the bound in Lemma \ref{lemma:noiseless-sum-rate}, the upper bound becomes tighter because for larger sum-rates the inner code grows in length which makes $\mu_{i,j}$ in \eqref{eq:mu} approximately optimal, the power lost in Lemma \ref{lemma:powerconstraint} negligible, and $SNR_{lb}(\tilde L)$ of Lemma \ref{lemma:snrbound} closer to $SNR(\tilde L)$.

\begin{lemma}
\label{lemma:blocklength-upper-bound}
Fix $\gamma,\beta \in (0,1)$ such that $-K\log(\beta) > \frac{1}{2}\log (1+P)$, and let $a_{lb}$ be defined as in Lemma~\ref{lemma:snrbound}. Assume noiseless feedback, i.e., $\sigma_n^2=0$. For any sum-rate $R$ such that
\begin{equation*}
\frac{1}{2}\log (1+P)<R<-K\log\beta,
\end{equation*}
let $L_0$ be the smallest integer $\tilde L$ such that
\begin{equation*}
\frac{K}{2(\tilde L+K-1)}\log\left (1+SNR_{lb}(\tilde L)\right) \geq R,
\end{equation*}
where $SNR_{lb}(\tilde L)$ is defined as in Lemma \ref{lemma:snrbound}.
Then
\begin{equation}
\label{eq:blocklength-ub}
L_0 \leq \left\lceil{\frac{-W(-\frac{a\ln2}{b}2^{-\frac{ac}{b}})}{a\ln2} - \frac{a}{b}}\right\rceil,
\end{equation}
where
\begin{equation*}
a = 2\left(\frac{R}{K}+\log\beta \right),
\end{equation*}
\begin{equation*}
b = a_{lb}(1-\gamma)2^{-2\frac{R}{K}(K-1)},
\end{equation*}
\begin{equation*}
c = \left[a_{lb}(1-\gamma)(K-1)+1\right]2^{-2\frac{R}{K}(K-1)},
\end{equation*}
and $W$ is the Lambert W function, i.e., $W(x)$ is the solution to $x=W(x)e^{W(x)}$. 
\end{lemma}
\begin{IEEEproof}
Define 
\begin{equation*}
R_{lb}(\tilde L) = \frac{K}{2(\tilde L+K-1)}\log \left(1+SNR_{lb}(\tilde L)\right),
\end{equation*}
where $SNR_{lb}(\tilde L)$ is defined as in Lemma \ref{lemma:snrbound}.

To derive the upper bound on $L_0$, we solve for $\tilde L$ that satisfies
\begin{equation*}
R_{lb}(\tilde L) = R.
\end{equation*}
After some manipulations, the preceding equation in $\tilde L$ reduces to
\begin{equation}
\label{eq:blocklength-eq}
2^{a \tilde L} = b\tilde L + c,
\end{equation}
which is known to have, by substitution, the term inside the ceil operator in \eqref{eq:blocklength-ub} as a solution in $\tilde L$. 

It can be easily shown that $R_{lb}(1) \leq \frac{1}{2}\log (1+P)$ and that $\lim_{\tilde L \rightarrow \infty} R_{lb}(\tilde L) = -K\log(\beta)$. Hence, there exists at least one $\tilde L$ such that $R_{lb}(\tilde L) = R$. Now, let us analyze \eqref{eq:blocklength-eq}. The left-hand side of the equation is a decreasing exponential function in $\tilde L$ because $a$ is negative. The right-hand side is a straight line in $\tilde L$ with a positive slope. Hence, \eqref{eq:blocklength-eq} can have one real valued solution only, call it $\hat L$. Then, $R_{lb}(\tilde L)\geq R$ for all $\tilde L \geq \hat L$. This validates the use of the ceil operater in \eqref{eq:blocklength-ub}. 

\end{IEEEproof}

\begin{corollary}
Let $f$ be defined as in Lemma \ref{lemma:f} and $\phi$ defined as in Lemma \ref{lemma:noiseless-sum-rate}. For any sum-rate $R$ such that
\begin{equation*}
\frac{1}{2}\log (1+P)<R<\frac{1}{2}\log (1+P\phi),
\end{equation*}
choose $\gamma \in (0,1)$ such that 
\begin{equation}
\label{eq:gamma}
\gamma >\gamma_{lb} = \frac{1}{P}f(2^{-\frac{R}{K}}).
\end{equation}
Choose $\beta$ such that 
\begin{equation}
\label{eq:beta}
\beta = f^{-1}(\gamma P),
\end{equation}
where $f^{-1}$ is the inverse of $f$.
For noiseless feedback (i.e., $\sigma_n^2=0$), let $L_0$ be the smallest $\tilde L > K$ such that 
\begin{equation}
\frac{K}{2(\tilde L+K-1)}\log\left (1+SNR^*(\tilde L)\right) \geq R,
\end{equation}
where given $\tilde L$, $SNR^*(\tilde L)$ is given by \eqref{eq:ortho-SNR} and that follows Lemma \ref{lemma_isit_1} and Lemma \ref{lemma_isit_2} using optimal $\gamma$ and $\beta$ values. Then $L_0$ can be upper bounded as follows
\begin{equation}
\label{eq:blocklength-upper-bound}
L_0 \leq \max\{K+1,\left\lceil{\frac{-W(-\frac{a\ln2}{b}2^{-\frac{ac}{b}})}{a\ln2} - \frac{a}{b}}\right\rceil\},
\end{equation}
where $a$, $b$, and $c$ are defined as in Lemma \ref{lemma:blocklength-upper-bound} with $\gamma$ and $\beta$ values chosen as in \eqref{eq:gamma} and \eqref{eq:beta}.
\end{corollary}
\begin{IEEEproof}
First, we choose $\gamma$ such that the linear coding scheme for the noiseless feedback case can achieve a sum-rate larger than $R$. To do so, we need
\begin{equation*}
-K\log{\beta}>R.
\end{equation*}
This implies
\begin{equation*}
\beta<2^{-\frac{R}{K}},
\end{equation*}
which also implies that
\begin{equation*}
f(\beta)>f(2^{-\frac{R}{K}}).
\end{equation*}
But for $\tilde L \rightarrow \infty$, the power constraint of the linear scheme reduces to $f(\beta) = \gamma P$. Then
\begin{equation*}
\gamma>\frac{1}{P}f(2^{-\frac{R}{K}}),
\end{equation*}
where the right-hand side is exactly $\gamma_{lb}$.

Now, for any $\gamma>\gamma_{lb}$, choosing $\beta = f^{-1}(\gamma P)$ satisfies the power constraint for any $\tilde L > K$ (Lemma \ref{lemma:powerconstraint}). The proof then follows by Lemma \ref{lemma:blocklength-upper-bound}. Note that the use of the $\max$ function in \eqref{eq:blocklength-upper-bound} function is to ensure that the upper bound on $L_0$  is no smaller than $K+1$. This is because of the way the linear scheme is constructed that requires $\tilde L>K$ for $R > \frac{1}{2}\log (1+P)$. By the discussion in the proof of Lemma \ref{lemma:blocklength-upper-bound}, larger blocklength is still a valid upper bound on $L_0$.

\end{IEEEproof}

In Fig. \ref{fig:blocklength}, we plot $\tilde L_{ub}$, which is the right-hand side of \eqref{eq:blocklength-upper-bound}, for sum-rates between $C_{nf} + 0.01\Delta$ and $C_{nf} + 0.9\Delta$, where $C_{nf} = \frac{1}{2}\log(1+\rho)$ and $\Delta =  \frac{1}{2}\log(1+\phi P) - \frac{1}{2}\log(1+P)$. We consider $P=10$ and $K=2$. For each sum-rate point, the $\gamma$ chosen was $\gamma = \gamma_{lb} + 0.2(1-\gamma_{lb})$. 

\begin{figure}[h]
\centering
 \includegraphics[width=0.7\textwidth]{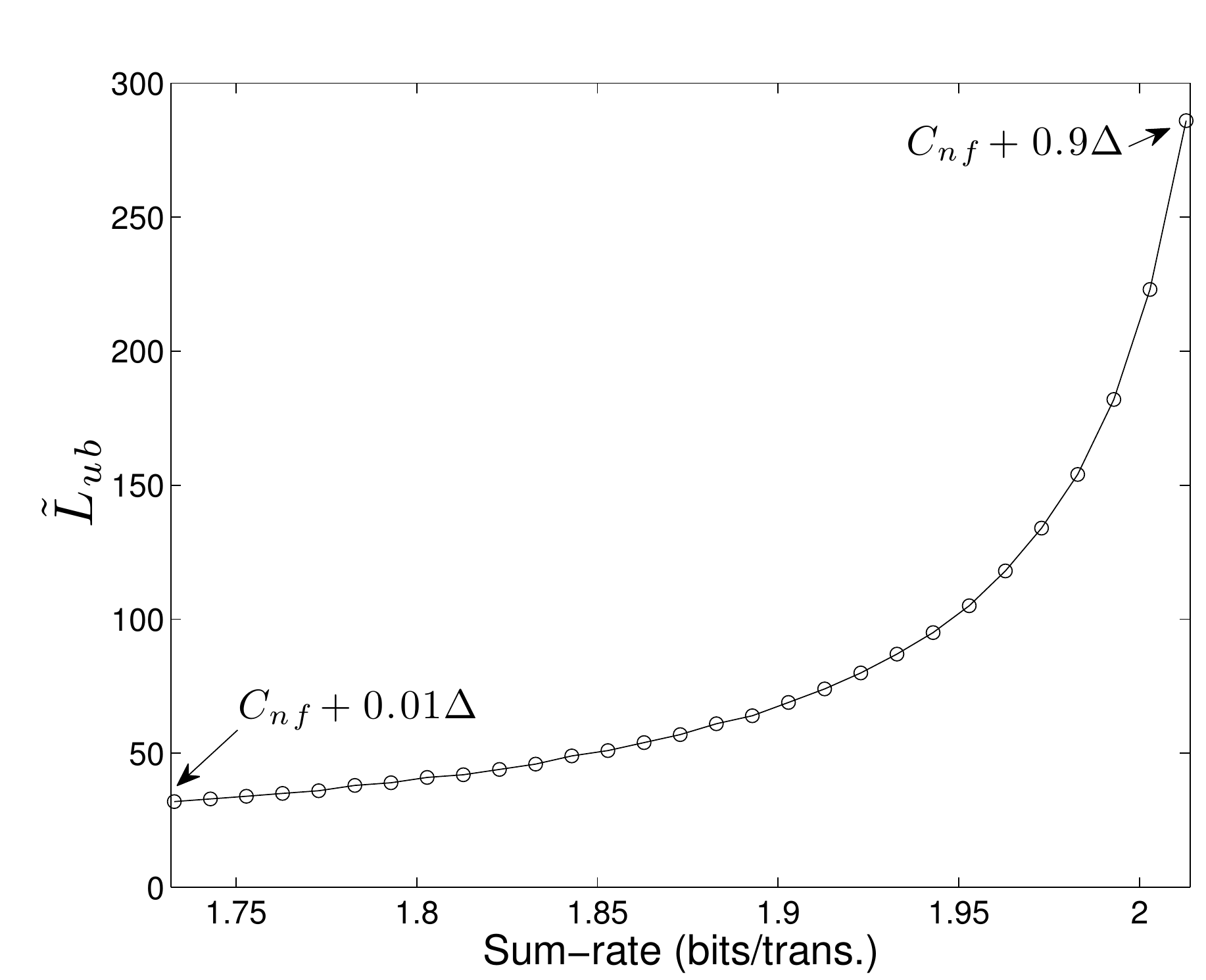}
\caption{Upper bound on the $\tilde L$ needed for the concatenated coding scheme to start to outperform a certain sum-rate for noiseless feedback. The values of the channel parameters are: $P=10$ and $K=2$.}
\label{fig:blocklength}
\end{figure}

\subsection{Sum-Rate Versus Feedback Noise Level}
In this section, we present, using computer experiments for numerical optimization, the achievable sum-rates given a certain feedback noise level. Specifically, we calculated the following
\begin{equation}
\label{eq:num-opt-equation}
R^* = \sup_{\substack{\tilde{L}\in\mathbb N \\ \beta \in (0,1) \\ \gamma \in [0,1]}} \frac{K}{2(\tilde{L}+K-1)}\log\left(1+SNR^*(\tilde{L},\beta,\gamma)\right),
\end{equation}
where given $\tilde L$, $\beta$ and $\gamma$, $SNR^*(\tilde{L},\beta,\gamma)$ is the SNR at any of the receivers given by \eqref{eq:ortho-SNR} and calculated using Lemma \ref{lemma_isit_1} and Lemma \ref{lemma_isit_2}.

\begin{figure}[h]
\centering
 \includegraphics[width=0.75\textwidth]{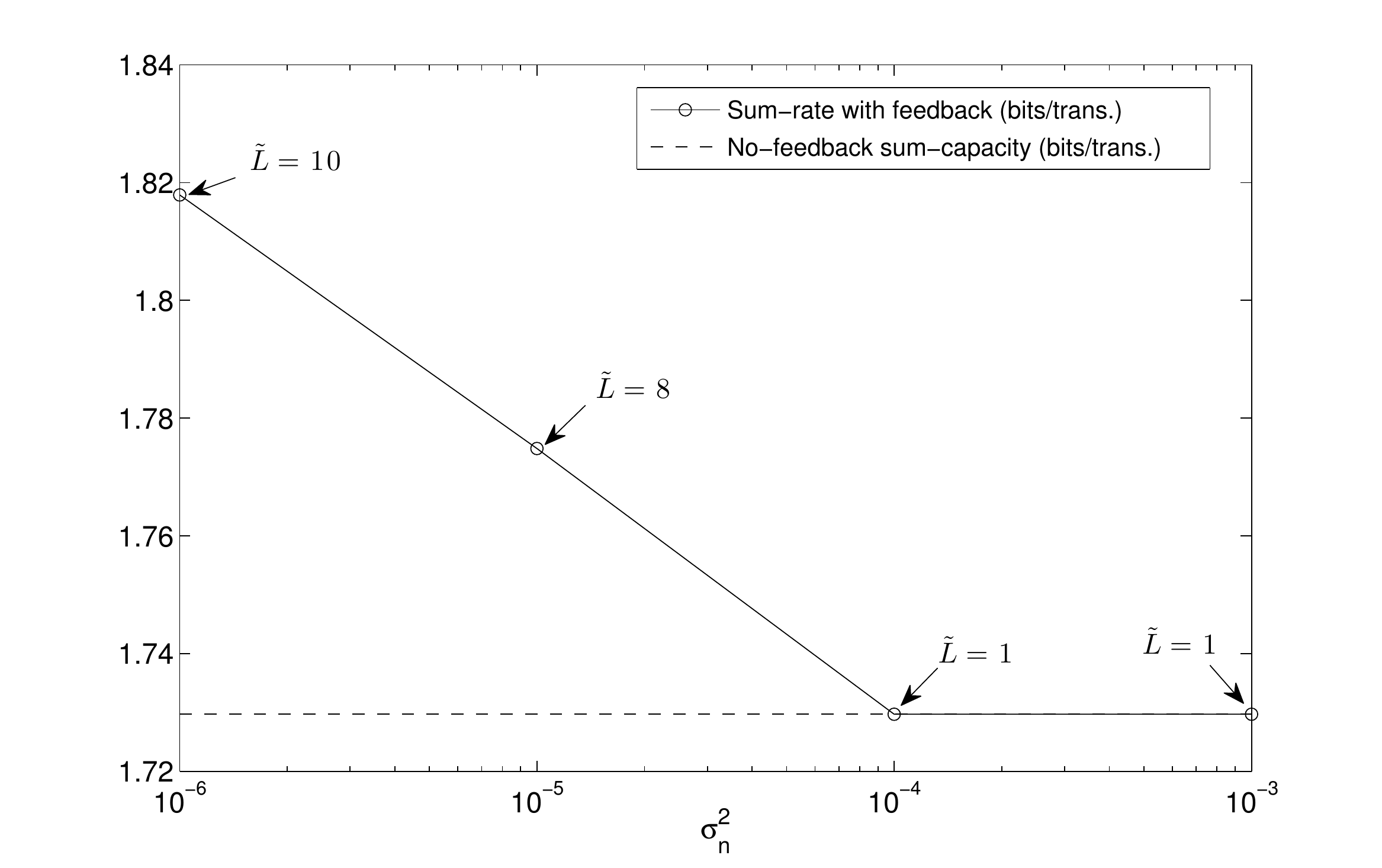}
\caption{Comparison between the sum-rates achievable by the proposed concatenated coding scheme and the no-feedback sum-capacity for $P=10$ and $K=2$. }
\label{fig:perpn}
\end{figure}

In Fig. \ref{fig:perpn}, we plot $R^*$ as a function of $\sigma_n^2$ for $P=10$ and $K=2$. The chosen points for $\sigma_n^2$ are $10^{-6}$, $10^{-5}$, $10^{-4}$, and $10^{-3}$. On the curve, the optimal $\tilde{L}$ for each  $\sigma_n^2$ is also shown. From the plot, we can see that for $\sigma_n^2=10^{-3}$ and $\sigma_n^2=10^{-4}$, the optimal $\tilde{L}$ is $1$, i.e., feedback is not utilized. (It is important to note that for the symmetric AWGN-BC orthogonal signaling is optimal for open-loop coding, which is encompassed by our scheme by having $\tilde{L}=1$ and $\gamma=0$). However, for $\sigma_n^2=10^{-5}$ and $\sigma_n^2=10^{-6}$, the concatenated coding scheme outperforms the no-feedback sum-capacity with optimal values for $\tilde{L}$ of $8$ and $10$, respectively. Note that as $\sigma_n^2\rightarrow0$, $R^*$ should approach the bound in Lemma \ref{lemma:noiseless-sum-rate} with the optimal $\tilde{L} \rightarrow \infty$. On the other hand, for all values of $\sigma_n^2$ greater than $10^{-3}$, the optimal $\tilde{L}$ should remain equal to 1 (with $\gamma=0$), at which open-loop coding outperforms the use of feedback information.


\section{Concatenated Coding for the Two-user AWGN-BC With One Noisy Feedback Link}
\label{sect:single-feedback}
In this section, we present a concatenated coding scheme for the two-user AWGN-BC with one noisy feedback link that uses the scheme presented in~\cite{bhaskaran}, which we will call the Bhaskaran scheme, with some modifications as an inner code. We will show that any rate tuple achieved by the Bhaskaran scheme for the noiseless feedback case, can be achieved by concatenated coding for the noisy feedback case if the noise variance in the feedback link is sufficiently small but not necessarily zero. 

The channel setup at hand is the same as in Section~\ref{sect:channel-setup}, but with $K=2$ and only one feedback link from one of the receivers. Without loss of generality, we will assume that reciever 1 has a feedback link to the transmitter and no feedback link from receiver 2. To follow the same channel description of Section~\ref{sect:channel-setup}, we can equivalently set $\sigma_{n_2}^2\mathrel{\mathop:}=\infty$ to render the feedback information from receiver 2 useless.

\subsection{Bhaskaran Scheme}
First, we start by a quick description of the original Bhaskaran scheme \cite{bhaskaran} that was designed for the noiseless feedback case. The transmitter forms two signals each intended to a respective receiver, and then transmitts the sum of the two signals. Let the signal intended to receiver 1 at time $\ell$ be $x_1[\ell]$ and that of receiver 2 be $x_2[\ell]$. Then, $x[\ell]=x_1[\ell]+x_2[\ell]$. 

For the receiver with the feedback link, which is assumed to be receiver 1, to form $x_1[\ell]$, the transmitter will use the linear feedback scheme presented in~\cite{merhav-weissman}, which is an extension of the S-K scheme~\cite{Schal1}, but for the Costa channel~\cite{dirtypaper} where $x_2[\ell]$ is considered to be the interfering signal and $z_1[\ell]$ is considered to be the noise. Assuming a fraction $\delta \in [0,1]$ of the power is allocated to $x_1[\ell]$, and let $P_1=\delta P$, then rates up to $R_1^{pf}$ are achievable to receiver 1, where
\begin{equation}
R_1^{pf} = \frac{1}{2}\log\left(1+\frac{P_1}{\sigma_{z_1}^2}\right).
\end{equation}

On the other hand, receiver 2 will have a fraction of the power $P_2=(1-\delta)P$ and will consider $x_1[\ell]$ as noise. Receiver 2 will ignore the first transmission, $x[1]$. By the structure of the S-K scheme, $x_1[2], x_1[3], \dots$ is a colored Gaussian process, hence the transmitter will form $x_2[\ell]$ as the ouput of an open loop coding scheme for the additive colored Gaussian noise channel, where the noise sequence is $\{x_1[\ell]+z_2[\ell]\}_{\ell>1}$. Using water-filling in the frequency domain as described in \cite{cover_book}, it is shown in \cite{bhaskaran} that any rate below $R_2^{pf}$ is achievable to receiver 2, where
\begin{equation}
R_2^{pf} = 
\begin{cases}
\int_0^\frac{1}{2} \log\left(\frac{2g(0)+P_2}{\tilde\sigma_{z_2}^2(f)}\right) df, & \text{if } 2g(0)+P_2 > \tilde\sigma_{z_2}^2(0) \\
\int_0^\frac{1}{2} \log\left(\frac{\tilde\sigma_{z_2}^2(a)}{\tilde\sigma_{z_2}^2(f)}\right) df, & \text{otherwise},
\end{cases}
\end{equation}
and 
\begin{equation}
\tilde\sigma_{z_2}^2(f) = \sigma_{z_2}^2 + \frac{P_1(\alpha^2-1)}{\alpha^2+1-2\alpha cos(2\pi f)},
\end{equation}
\begin{equation}
g(x)=\int_x^\frac{1}{2} \log\left(\tilde\sigma_{z_2}^2(f)\right) df,
\end{equation}
$\alpha = \sqrt{1+\nicefrac{P_1}{\sigma_{z_1}^2}}$ and $a$ is solution of $(1-2a)\tilde\sigma_{z_2}^2(a)-2g(a)=P_2$ \cite{bhaskaran}.

\subsection{Noisy-Bhaskran Scheme}
We discuss here some modifications on the Bhaskaran scheme \cite{bhaskaran} to accomodate the presence of noise in the feedback link. We will call the modified scheme Noisy-Bhaskaran. The necessary modifications are the following:
\begin{enumerate}
\item The transmitter in the original Bhaskaran scheme forms $x_1[\ell]$ as a linear combination of $z_1[1],z_1[2],\dots,z_1[\ell-1]$ for $\ell\geq2$. For the noisy feedback case, the transmitter does not know $z_1[1],z_1[2],\dots,z_1[\ell-1]$, however it has knowledge of $z_1[1]+n_1[1],z_1[2]+n_1[2],\dots,z_1[\ell-1]+n_1[\ell-1]$. We will assume that the transmitter uses the sequence $z_1[1]+n_1[1],z_1[2]+n_1[2],\dots,z_1[\ell-1]+n_1[\ell-1]$ thinking it is $z_1[1],z_1[2],\dots,z_1[\ell-1]$, and for forming the scaling coefficients uses $\sigma_{z_1}^2+\sigma_{n_1}^2$ instead of $\sigma_{z_1}^2$. Another way to think of this, is that the transmitter will be forming $x_1[\ell]$ extacly as if the channel at hand was of forward noise $z_1[\ell]+n_1[\ell]$ to receiver 1 and of noiseless feedback. Receiver 1 will form the estimate of the message point $\theta$ exactly as in the original Bhaskaran scheme assuming the transmitter is operating for noiseless feedback.
\item For receiver 2, following the previous step the sequence $x_1[2],x_1[3],\dots$ is still a Gaussian process whose covariance matrix is as described in \cite{bhaskaran} but with $\sigma_{z_1}^2$ replaced by $\sigma_{z_1}^2+\sigma_{n_1}^2$.
\end{enumerate}

For the receiver with feedback, the message is mapped to a parameter $\theta$ for linear coding. Since for receiver 2 we are using open loop coding, the tranmsitter decides on a codeword corresponding to the message, call it $W_2$, intended to receiver 2 before starting transmission. Hence, $x_2[1], x_2[2], \dots,$ and $ x_2[L]$ are known to the transmitter before tranmission. In Bhaskaran scheme, as in~\cite{merhav-weissman}, the transmitter forms $x_1[\ell]$ exactly as in the S-K scheme except that interference is subtracted in the first transmission. 
We will now follow a similar vector representation as Section~\ref{sect:linear-general} for receiver 1 by assuming that interference from $x_2[\ell]$ is not present.  Let $\hat\theta$ be the estimate of $\theta$ at receiver 1, then, and similar to~\eqref{eq:thetahat}, we can write
\begin{equation}
\label{eq:single-user-theta}
 \hat\theta = \mathbf q_1^T\mathbf g_1\theta +\mathbf q_1^T(\mathbf I+ \mathbf F_1)\mathbf z_1 + \mathbf q_1^T\mathbf F_1\mathbf n_1,
\end{equation}
where $\mathbf I$ is the idendity matrix. 
The receive SNR at receiver 1 can be written as
\begin{equation}
\label{eq:single-user-snr}
SNR(L,\sigma_{n_1}^2) = \frac{(\mathbf q_1^T\mathbf g_1)^2E[\theta^2]}{\sigma_{z_1}^2\|\mathbf q_1^T(\mathbf I+ \mathbf F_1) \|^2 + \sigma_{n_1}^2\|\mathbf q_1^T\mathbf F_1\|^2},
\end{equation}
where the dependence of the SNR on $L$ and $\sigma_{n_1}^2$ was made explicit. Note that the second argument of $SNR(L,\sigma_{n_1}^2)$ only captures $\sigma_{n_1}^2$ that explicitly appears in \eqref{eq:single-user-snr}, i.e., it does not capture the possible dependence of $\mathbf q_1$, $\mathbf g_1$, or $\mathbf F_1$ on $\sigma_{n_1}^2$. 

We will assume that the power spent for interference subtraction in the first transmission will be taken out from the power allocated to $x_1[k]$. For blocklength of $L$, the total power available to $x_1[\ell]$ is $L P_1$. Assume that the power spent for interference subtraction is $\delta_{IS}(L) L P_1$, where $\delta_{IS}(L)$ is a function of $L$ with range $[0,1]$. Although $\delta_{IS}(L)$ may have to be larger than 1 for small $L$, for our purposes we will set $\delta_{IS}(L)=1$ when interference substraction requires $\delta_{IS}(L)>1$, which we will only happen for small $L$ because, and as discussed in~\cite{merhav-weissman} and~\cite{bhaskaran},  $\delta_{IS}(L) L P_1 \rightarrow 0$ as $L \rightarrow \infty$. Now, we can write the power constraint on the feedback scheme as such
\begin{equation}
\label{eq:single-user-pwrconstraint}
 \mathbf g_1^T \mathbf g_1 E[\theta^2]+ (\sigma_{z_1}^2+\sigma_{n_1}^2) \|\mathbf F_1\|_F^2 \leq  L (1-\delta_{IS}(L)) P_1.
\end{equation}

Finally, we like to note that constructing $\mathbf g_1$, $\mathbf q_1,$ and $\mathbf F_1$ as in the Bhaskaran scheme, it can be shown that
\begin{equation}
\lim_{L\rightarrow \infty} \frac{1}{2L}\log\left(1+ SNR(L,0)\right) = R_1^{pf}.
\end{equation}

\subsection{Concatenated Coding Scheme}
\label{sect:single-feedback-concatenated}
The concatenated coding scheme we will present here is similar to the scheme described in Section~\ref{sect:concatenated-general} with slight modification to accomodate the use of open loop coding in the inner code. 

Consider that we are using the Noisy-Bhaskaran scheme for a finite blocklength of $L$. From \eqref{eq:single-user-theta}, we observe that for finite blocklength $L$, the stochastic relation between $\theta$ and $\hat\theta$ can be modeled as an effective scalar AWGN channel without feedback whose input is $\theta$ and output is $\hat\theta$. The SNR of this effective channel, which in this case is a scalar AWGN channel without feedback, is given by~\eqref{eq:single-user-snr}. Similar to Section~\ref{sect:concatenated-general}, using open-loop coding for the AWGN channel to code over the latter effective channel, we can achieve any rate $R_1\geq0$ to receiver 1 satisfying
\begin{equation}
R_1<\frac{1}{2L}\log\left(1+SNR(L,\sigma_{n_1}^2)\right),
\end{equation}
where $SNR(L,\sigma_{n_1}^2)$ is as defined in~\eqref{eq:single-user-snr}. Note that if \eqref{eq:single-user-pwrconstraint} is satisfied by the Noisy-Bhaskaran scheme for blocklength of $L$, then the overall code (i.e., with open-loop coding) satisfies the average power constraint $P_1$ of receiver 1.

Now, assume that for the open loop code of receiver 2, the codewords are of length $L$ and the codebook is of size $2^{LR'_2}$, where $R'_2 \in [0,\infty]$. For convenience, we will assume that $2^{LR'_2}$ is an integer. Note that all the codewords of the codebook have their first entry equal to zero. Assume that the message $W_2$ intended to receiver 2 is in $\{w_1,w_2,\dots,w_{2^{LR'_2}}\}$. Let the decision of the decoder at receiver 2 be $\hat W_2$ whose range is $\{w_1,w_2,\dots,w_{2^{LR'_2}}\}$. The stochastic relation between $W_2$ and $\hat W_2$ can be modeled as a discrete memoryless channel (DMC) with input and output alphabet $\{w_1,w_2,\dots,w_{2^{LR'_2}}\}$ and transitional probabilities given by
\begin{equation}
p(w_i|w_j)=Pr\{\hat W_2=w_i| W_2=w_j\},
\end{equation}
 where $i,j \in \{1,2,\dots,2^{LR'_2}\}$.
Thus, if we use an open loop encoder for coding over the latter effective DMC, we can achieve any rate $R_2\geq0$ to receiver 2 if 
\begin{equation}
\label{eq:rates-rx2}
R_2<\frac{1}{L}\max_{p_{W_2}}I(W_2;\hat W_2),
\end{equation}
where $p_{W_2}$ is the probability mass function of $W_2$, and $I(W_2;\hat W_2)$ is the average mutual information between $W_2$ and $\hat W_2$.

\begin{theorem}
\label{thm:single-feedback}
For any rate tuple $(R_1,R_2)$ such that $R_1<R_1^{pf}$ and $R_2<R_2^{pf}$, there exists $\epsilon>0$ such that $(R_1,R_2)$ is achievable by the  concatenated coding scheme over the AWGN-BC with a single noisy feedback link from receiver 1 with feedback noise variance $\sigma_{n_1}^2$ as large as $\epsilon$.
\end{theorem}

\begin{IEEEproof}
Choose $\epsilon_1>0$ such that 
\begin{equation}
R_2 < R_2^{pf}\left(\frac{P_2}{\sigma_{z_2}^2},\sigma_{z_1}^2+\epsilon_1 \right)\leq R_2^{pf}\left(\frac{P_2}{\sigma_{z_2}^2},\sigma_{z_1}^2\right),
\end{equation}
where the dependence of $R_2^{pf}$ on $\nicefrac{P_1}{\sigma_{z_2}^2}$ and $\sigma_{z_1}^2$ was made explicit.

For the Bhaskaran scheme designed for a channel similar to the given channel but with forward noise variance to receiver 1 of $\sigma_{z_1}^2+\epsilon_1$ instead of $\sigma_{z_1}^2$, we fix a sequence of codes for receiver 2 that achieves $R'_2$, where $R'_2$ is such that
\begin{equation}
R_2<R'_2<R_2^{pf}\left(\frac{P_2}{\sigma_{z_2}^2},\sigma_{z_1}^2+\epsilon_1\right).
\end{equation}
 Let the capacity of the effective DMC for each code of this sequence and that has blocklength $L$ be $\max_{p_{W_2}}\tilde{I}(W_2(L);\hat W_2(L))$, then
\begin{equation}
\label{eq:limit-MI}
\lim_{L\rightarrow \infty} \frac{1}{L}\max_{p_{W_2}}\tilde{I}(W_2(L);\hat W_2(L)) > R_2,
\end{equation}
where the dependence of $W_2$ and $\hat W_2$ on $L$ was made explicit. For convenience, we assume the limit in \eqref{eq:limit-MI} exists. If the limit does not exist, limit superior can be used instead and the proof will require very small changes to accomodate that. 

Choose $L_1$ such that 
\begin{equation}
R_1<\frac{1}{2}\log \left(1+\frac{(1-\delta_{IS}(L_1))P_1}{\sigma_{z_1}^2}\right),
\end{equation}
where $\delta_{IS}$ corresponds to substracting interference from the sequence of codes we have just fixed.
Using the S-K scheme but for $(1-\delta_{IS}(L_1))P_1$ power available to receiver 1 instead of $P_1$, we have
\begin{equation}
\label{eq:snr-proof}
\lim_{L\rightarrow \infty} \frac{1}{2L}\log\left(1+ \overline{SNR}(L,0)\right) = \frac{1}{2}\log\left(1+\frac{(1-\delta_{IS}(L_1))P_1}{\sigma_{z_1}^2}\right),
\end{equation}
where $\overline{SNR}$ here is given by~\eqref{eq:single-user-snr} and its $\mathbf g_1$, $\mathbf F_1$, and $\mathbf q_1$ matrices are constructed according to the S-K scheme that is designed for power constraint of $(1-\delta_{IS}(L_1))P_1$ .

Now, choose $L_0$ such that 
\begin{enumerate}
\item $\delta_{IS}(L_0) \leq \delta_{IS}(L_1)$
\item $R_1 < \frac{1}{2L_0}\log\left(1+ \overline{SNR}(L_0,0)\right)$
\item $R_2 < \frac{1}{L_0}\max_{p_{W_2}}\tilde{I}(W_2(L_0);\hat W_2(L_0))$.
\end{enumerate}
To find such $L_0$, we find an $L$ that satisfies each of three the conditions separately and then choose the largest among them. Specifically,
\begin{enumerate}
\item $\delta_{IS}(L)$ is monotonically descreasing in $L$ and so any $L_0\geq L_1$ suffice. Let our choice be $L_0^{(1)}$.
\item By \eqref{eq:snr-proof} and by the definition of the limit, there exits $L_0^{(2)}$ such that for any $L\geq L_0^{(2)}$ we have $R_1 < \frac{1}{2L}\log\left(1+ \overline{SNR}(L,0)\right)$.
\item By \eqref{eq:limit-MI} and by the definition of the limit, there exits $L_0^{(3)}$ such that for any $L\geq L_0^{(3)}$ we have $R_2 < \frac{1}{L}\max_{p_{W_2}}\tilde{I}(W_2(L);\hat W_2(L))$.
\end{enumerate}
Then, $L_0=\max\{L_0^{(1)},L_0^{(2)},L_0^{(3)}\}$ would satisfy the three conditions together.

Let $\mathbf g_1^{(L_0)}$, $\mathbf F_1^{(L_0)}$, and $\mathbf q_1^{(L_0)}$ be the matrices of the S-K scheme we are using but for blocklength $L_0$. Note that $\mathbf g_1^{(L_0)}$ and $\mathbf F_1^{(L_0)}$ satisfy
\begin{equation}
{\mathbf g_1^{(L_0)}}^T \mathbf g_1^{(L_0)} E[\theta^2]+ \sigma_{z_1}^2 \|\mathbf F_1^{(L_0)}\|_F^2 \leq  L_0 (1-\delta_{IS}(L_1)) P_1.
\end{equation}
Choose $0<\epsilon_2\leq \epsilon_1$ and $\mathbf g_1'$ (also with the only non-zero entry in the first position) such that
\begin{equation}
\label{eq:powerconstraint_epsilon2}
\mathbf g_1'^T \mathbf g_1' E[\theta^2]+ (\sigma_{z_1}^2+\epsilon_2) \|\mathbf F_1^{(L_0)}\|_F^2 \leq  L_0 (1-\delta_{IS}(L_1)) P_1
\end{equation}
and
\begin{equation}
\label{eq:snr_epsilon2}
R_1 < \frac{1}{2L_0}\log\left(1+ SNR'(L_0,\epsilon_2)\right),
\end{equation}
where $SNR'$ is the same as $\overline{SNR}$ but with $\mathbf g_1^{(L_0)}$ replaced with $\mathbf g_1'$.
It can be shown that such  $\mathbf g_1'$ and $\epsilon_2$ exist by a similar argument as in the proof of Theorem 1. Let $\mathbf F_1'$ be of construction similar to $\mathbf F_1^{(L_0)}$ but with $\sigma_{z_2}^2$ replaced with $\sigma_{z_2}^2 + \epsilon$  in its construction, where $0<\epsilon\leq\epsilon_2$ is such that
 \begin{equation}
R_1 < \frac{1}{2L_0}\log\left(1+ SNR''(L_0,\epsilon_2)\right),
\end{equation}
and $SNR''$ is the same as $SNR'$ but with $\mathbf F_1^{(L_0)}$ replaced with $\mathbf F_1'$.
Since $\|\mathbf F_1^{(L_0)}\|_F^2 \geq \|\mathbf F_1'\|_F^2$ (by the construction of the S-K scheme) and $0<\epsilon\leq\epsilon_2$, we have
\begin{equation}
\mathbf g_1'^T \mathbf g_1' E[\theta^2]+ (\sigma_{z_1}^2 + \epsilon)\|\mathbf F_1'\|_F^2 \leq  L_0 (1-\delta_{IS}(L_1)) P_1
\end{equation}
and
\begin{equation}
R_1 < \frac{1}{2L_0}\log\left(1+ SNR''(L_0,\epsilon)\right).
\end{equation}



For the Noisy-Bhaskaran scheme of blocklength $L_0$ and over the given channel but with $\sigma^2_{n_1}=\epsilon>0$, we have found
\begin{itemize}
\item For receiver 1: $\mathbf F_1'$, $\mathbf g_1'$, and $\mathbf q_1^{(L_0)}$ such that
\begin{equation}
 R_1 < \frac{1}{2L_0}\log\left(1+ SNR''(L_0,\epsilon)\right)
\end{equation}
and
\begin{align}
 \mathbf g_1'^T \mathbf g_1' E[\theta^2]+ (\sigma_{z_1}^2 + \epsilon)\|\mathbf F_1'\|_F^2 & \leq  L_0 (1-\delta_{IS}(L_1)) P_1 \\
&\leq L_0(1-\delta_{IS}(L_0)) P_1.
\end{align}
\item For receiver 2: a code of blocklength $L_0$ that satisfies
\begin{equation}
R_2 < \frac{1}{L_0}\max_{p_{W_2}}I(W_2(L_0);\hat W_2(L_0))
\end{equation}
for the case of forward noise variance to receiver 1 of $\sigma^2_{z_1}+\epsilon_1$ that reqiures no larger than $\delta_{IS}(L_0)L_0P$ power to be subtracted. Hence, there exists a code of length $L_0$ for the case of $\sigma^2_{z_1}+\epsilon \leq \sigma^2_{z_1}+\epsilon_1$ that requires no larger than $\delta_{IS}(L_0)L_0P$ power for interference subtraction and is such that 
\begin{equation}
R_2 < \frac{1}{L_0}\max_{p_{W_2}}I'(W_2(L_0);\hat W_2(L_0)),
\end{equation}
where $\max_{p_{W_2}}I'(W_2(L_0);\hat W_2(L_0))$ is the capacity of the effective DMC of the new code for the case of $\sigma^2_{z_1}+\epsilon$.
\end{itemize}

Therefore, by using concatenated coding as presented in Section \ref{sect:single-feedback-concatenated} over the Noisy-Bhaskaran scheme of blocklength $L_0$ just described, the rate tuple $(R_1,R_2)$ is achievable for $\sigma_{n_1}^2$ as large as $\epsilon$.
\end{IEEEproof}

In \cite{ramji}, the same channel was considered, and in particular the symmetric case. For high forward channel SNR, the scheme in \cite{ramji} showed improvements on the no-feedback sum-capacity for feedback noise level as large as forward noise level. However, for low, but still practical, forward channel SNR, the scheme in \cite{ramji} shows negligible improvement on the no-feedback sum-capacity even for the noiseless feedback case. The result of Theorem \ref{thm:single-feedback} is an improvement on that, albeit for small feedback noise level.
\section{Conclusion}
\label{sect:conclusion}
In this paper, we have used a concatenated coding design that uses linear feedback schemes as inner codes to achieve rate tuples for the $K$-user AWGN-BC with noisy feedback outside the no-feedback capacity region. We have shown an achievable rate region of linear feedback schemes for the noiseless feedback case to be achievable by the concatenated coding scheme for sufficiently small feedback noise level. We also presented a linear feedback scheme for the symmetric $K$-user AWGN-BC with noisy feedback that was used as an inner code in the concatenated coding scheme that was itself optimized to achieve sum-rates above the no-feedback sum-capacity. The concatenated coding design was also applied to the two-user AWGN-BC with a single noisy feedback link from one of the receivers.

\bibliographystyle{IEEEtran}
\bibliography{all}

\end{document}